\documentclass[12pt]{article}
\usepackage{xcolor}
\usepackage[utf8]{inputenc}
\usepackage{algorithm,bm,algorithmic}
\usepackage{lmodern}
\usepackage{amsmath,amssymb,amsthm}
\usepackage{comment}
\usepackage{rotating}

\usepackage{natbib}
\usepackage{indentfirst}
\usepackage[margin=1in,top=1in,
            nohead,
            pdftex]{geometry}
\usepackage{url}
\usepackage{graphicx}
\usepackage{subcaption}
\usepackage[colorlinks=true,citecolor=blue,urlcolor=blue,linkcolor=blue]{hyperref}
\usepackage{booktabs}

\def\u{{\boldsymbol{u}}}
\def\e{{\boldsymbol{e}}}

\def\bmu{{\boldsymbol{\mu}}}

\def\A{{\mathbf{A}}}
\def\I{{\mathbf{I}}}
\def\S{{\mathbf{S}}}
\def\V{{\mathbf{V}}}
\def\X{{\mathbf{X}}}
\def\W{{\mathbf{W}}}
\def\J{{\mathbf{J}}}

\def\bLambda{{\mathbf{\Lambda}}}
\def\bSigma{{\mathbf{\Sigma}}}
\def\bOmega{{\mathbf{\Omega}}}
\def\bGamma{{\mathbf{\Gamma}}}
\def\bDelta{{\mathbf{\Delta}}}

\def\T{{\mathcal{T}}}

\def\tr{{\rm tr}}
\def\pr{\mathbb{P}}
\def\E{\mathbb{E}}
\def\var{\mathrm{Var}}

\def\vec{\mathrm{vec}}

\DeclareMathOperator*{\argmin}{arg\,min}

\newcommand{\red}[1]

\newtheorem{theorem}{Theorem}
\newtheorem{definition}{Definition}

\newtheorem{assumption}{Assumption}
\newtheorem{lemma}{Lemma}
\numberwithin{equation}{section}

\newtheorem{remark}{Remark}

\title{Tensor Elliptical Graphic Model}
\author{Jixuan Liu$^1$, Zhengke Lu$^1$, Le Zhou$^2$ and Long Feng$^1$ and Zhaojun Wang$^1$\\
$^1$Nankai University and $^2$Hongkong Baptist University}
\date{\today}

\begin{document}

\maketitle
\begin{abstract}
     We address the problem of robust estimation of sparse high dimensional tensor elliptical graphical model. Most of the research focus on tensor graphical model under normality. To extend the tensor graphical model to more heavy-tailed scenarios, motivated by the fact that up to a constant, the spatial-sign covariance matrix can approximate the true covariance matrix when the dimension turns to infinity under tensor elliptical distribution, we proposed a spatial-sign-based estimator to robustly estimate tensor elliptical graphical model, the rate of which matches the existing rate under normality for a wider family of distribution, i.e. elliptical distribution. We also conducted extensive simulations and real data applications to illustrate the practical utility of the proposed methods, especially under heavy-tailed distribution.

\end{abstract}
\section{Introduction}
Tensor-valued data, also known as multi-dimensional data, arise in a wide range of scientific fields, such as neuroimaging\citep{zhou2013tensor,karahan2015tensor}, genomics\citep{hore2016tensor} and economics\citep{chen2023statistical}. As the complexity of data increases in modern applications, tensor-based methods are receiving increasing attention for their ability to preserve and utilize the inherent multi-directional structure of the data, rather than collapsing it into vectors and risking the loss of structural information and the development of grossly imbalanced sample sizes and data dimensions. A variety of statistical techniques have been proposed for analyzing tensor-valued data, for example, Gaussian graphical models for capturing conditional dependencies across modes\citep{sun+wang+liu+cheng-2015-Tlasso,lyu2019tensor,Min02012022},  {tensor-based classification and clustering} \citep{lyu2017discriminating,pan2019covariate,Luo2022,han2022exact,Mai02102022,wang2024,hou2024}, low-rank tensor decompositions and dimension reduction\citep{li2017parsimonious,zhang2018tensorsvd,zhang2019optimalavd,wangning2022,Han2023,Li2025}, and tensor regression models\citep{zhou2013tensor,hao2021sparsereg,zhou2023partially,Luo2024}.

However, most of the research focus on the tensor-valued data under normality, particularly in  estimating precision matrices that characterize conditional dependence structures across modes. In recent years, the estimation of sparse precision matrices under Kronecker-structured covariance has attracted increasing interest, with developments extending from matrix-valued data to tensor-valued data. For Gaussian graphical models, \cite{sun+wang+liu+cheng-2015-Tlasso} proposed a cyclic estimation method based on a non-convex optimization problem. \cite{Min02012022} developed a fast and separable estimation approach, while \cite{xu2017gd} introduced a gradient descent based method to further improve computational efficiency and estimation accuracy by simulations.

While the normality assumption contributes to theoretical analysis and methodological development, it is hard to satisfy in practice. For instance, neuroimaging data frequently exhibit heavy tails or other forms of distributional heterogeneity that deviate from normality\citep{wang2023highreg_tdist}. As a result, methods that rely heavily on the normality assumption may yield misleading conclusions or suffer from reduced robustness in real-data applications.

In this paper, we focus on the problem of estimating the precision matrices given a sequence of random $K$th-order tensors,  {following} the tensor elliptical distribution\citep{arashi-2023-tensor_elliptical}. Tensor elliptical distributions are a extension of the elliptical distribution in multivariate case and offer a flexible extension of the multivariate normal family. As a class of symmetric distributions that allow for heavier or lighter tails, they provide a more realistic framework for modeling complex data structures encountered in practice. In Gaussian graphical model, the sparse precision matrix of each way characterizes the conditional independence among the unfolded tensor data. In non-Gaussian graphical models, the precision matrices similarly capture conditional dependencies within each mode, although the interpretation may vary depending on the underlying distributional assumptions. Beyond graphical modeling and structure recovery, precision matrices also play a key role in tensor classification\citep{pan2019covariate,chen2024high}. Developing more robust estimators of precision matrices—particularly under non-Gaussian or heavy-tailed distributions—can significantly enhance classification accuracy and stability in high-dimensional applications\citep{lu2025robust}.

The spatial sign is a commonly used nonparametric tool for robust inference and has demonstrated excellent performance under heavy-tailed distributions. For vector-valued data, spatial sign–based methods have been extensively studied and applied to high-dimensional one-sample and two-sample testing\citep{zou2014multivariate,feng2016multivariate,cheng2023statistical}, principal component analysis\citep{feng-2024-SSPCA}, elliptical graphical model recovery and discriminant analysis\citep{lu2025robust}. These approaches exhibit favorable theoretical and empirical properties under non-Gaussian noise. In this paper, we extend the spatial sign to the estimation of precision matrices for tensor-valued data, and develop a spatial-sign based fast and efficient estimation procedure that is robust to heavy-tailed data while preserving the inherent tensor structure. The contributions of this paper are outlined as follows:
\begin{itemize}
    \item[(i)]We propose a spatial-sign based method  {for estimating precision matrices of tensor data}, leveraging spatial signs which are well-established tools for robust analysis  {of} heavy-tailed data. Our new method not only provides superior  performance under heavy-tailed distributions but also remains competitive with mean-based methods when the data are normally distributed.
    \item[(ii)]The proposed method estimates the precision matrix for each mode separately, without iterative updates across modes. The entire estimation procedure is completed in a single step, which substantially reduces computational cost, especially in high-dimensional settings. This leads to the method being faster and more efficient.
\end{itemize}

The remainder of the paper is organized as follows. Section 2 introduces the tensor graphical model along with the corresponding estimation methodology. Section 3 presents simulation studies to evaluate the performance of the proposed approach. Section 4 demonstrates the application of the method to a real dataset. All theoretical proofs are provided in the Appendix.

\textbf{Notations:} For vector $\boldsymbol x$, we use the notation $\Vert \boldsymbol x\Vert_2$ and $\Vert \boldsymbol x\Vert_\infty$ to denote its Euclidean norm and maximum-norm respectively. For matrix $\A=(A_{i,j})\in \mathbb R^{d\times d }$, we denote $\Vert\A\Vert_\infty$, $\Vert\A\Vert_2$ and $\Vert\A\Vert_F$ as its max, spectral and Frobenius norm respectively. We define $\Vert \A\Vert_{L_1}=\max_{1\leq j\leq d}\sum_{i=1}^d\vert A_{i,j}\vert$ and $\|\A\|_{L_\infty}=\max_{1\leq i\leq d}\sum_{j=1}^d\vert A_{i,j}\vert$ as its $L_1$ and $L_\infty$ norm respectively and $\Vert\A\Vert_{1,\text{off}}=\sum_{i\not= j}\vert A_{i,j}\vert$ as its off-diagonal $l_1$ norm. Denote $\vec(\A)$ as the vectorization of $\A$ which stacks the columns of $\A$, $\tr(\A)$ be the trace of $\A$ and $[\A]_{i,j}$ be the element in $(i,j)$, $i,j\in\{1,\cdots,d\}$. We follow the tensor notations in \cite{kolda+Bader-2009-tensor}. The norm of a tensor $\T$ is defined as $\Vert\T\Vert=\{\sum_{i_1=1}^{I_1}\cdots\sum_{i_K=1}^{I_K} t_{i_1,\ldots,i_K}^2\}^{1/2}$. In addition, for a list of matrices $\{\A_1,\ldots,\A_K\}$ with $\A_k\in \mathbb R^{m_k\times m_k}$, $k=1,\ldots,K$, we define $\T\times \{\A_1,\ldots,\A_K\}:=\T\times_1\A_1\times_2\ldots\times_K \A_K$. Let $\T^{\boldsymbol p}$ denote the space of all tensors $ \mathcal{X}$, where $\mathcal{X}$ is a tensor of order $k$, $\boldsymbol p = (p_1,\ldots,p_k)$. Let the space of matrix $\T^{\boldsymbol p}_\otimes=\{\X:\X=\X_1\otimes\ldots\otimes\X_k,\X_i\in \mathbb R^{p_i\times p_i} \}$. For any random variable $x \in \mathbb{R}$, we define the sub-Gaussian and sub-exponential norms of $x$ as $\|x\|_{\psi_2} = \sup_{k \geq 1} k^{-1/2} ( \mathbb{E}|x|^k )^{1/k}$ and $\|x\|_{\psi_1} = \sup_{k \geq 1} k^{-1} ( \mathbb{E}|x|^k )^{1/k}$, respectively. 

\section{Tensor Elliptical Graphic Model}
Section~\ref{subsec:model} introduces the tensor elliptical graphical model. In Section~\ref{subsec:algorithm}, we present an estimation method for precision matrices across modes, and the theoretical property of the proposed estimators is established in Section~\ref{subsec:theorem}.

\subsection{Model settings}\label{subsec:model}

In this paper, we consider a sequence of random $K$th-order tensors, that is, $\T_i\in \mathbb R^{p_1\times\cdots p_K},i=1,\ldots,n$,  {each following} the tensor elliptical distribution. The tensor elliptical (TE) distribution is assumed as follows. 

Let $\u^{(p^*)}$, $p^*=\prod_{i=1}^K p_i$ denote a random vector distributed uniformly on the unit sphere surface in $\mathbb R^{p^*}$
. A random tensor $\boldsymbol{\mathcal { X }}$ is TE of order $K$, denoted by $\boldsymbol{\mathcal { X }} \sim \mathcal{E}_{\boldsymbol{p}}(\boldsymbol{\mu}, \mathcal{S}^*_{\Sigma}, \psi)$, if
\begin{equation*}
\boldsymbol{x}=\operatorname{vec}(\boldsymbol{\mathcal { X }})=\bmu+v (\bSigma^*)^{\frac{1}{2}} \u^{\left(p^*\right)},
\end{equation*}
where $\boldsymbol{x}, \boldsymbol{\mu} \in \mathcal{T}^{\boldsymbol{p}}$, $\mathcal{S}^*_{\Sigma}=\{\bSigma_1^*,\cdots,\bSigma_K^*\}$, $\bSigma^*=\bSigma_1^*\otimes\cdots\otimes\bSigma_K^*$, $\bSigma_k^*$, $k=1,\ldots,K$ are symmetric and invertible, $\boldsymbol{p}=\left(p_1, \ldots, p_K\right)$, a random variable $v \geq 0$ is independent of $\boldsymbol{u}^{\left(p^*\right)}$, and $v \sim F$,  {where $F(\cdot)$ is some cumulative distribution function (cdf) over $[0, \infty)$.}

The pdf of the TE distribution is given by
$$
f_{\mathcal{X}}(\boldsymbol{x})=|\bSigma^*|^{-\frac{1}{2}} g\left((\boldsymbol{x}-\boldsymbol{\mu})^{\prime} (\bSigma^*)^{-1}(\boldsymbol{x}-\boldsymbol{\mu})\right),
$$
where $g(\cdot)$ is a non-negative function(density generator) satisfying

$$
\int_{\mathbb{R}^{+}} y^{\frac{1}{2} p^*-1} g(y) \mathrm{d} y<\infty.
$$

For further discussions and properties of TE distributions, see \citet{arashi-2023-tensor_elliptical}.




\subsection{Algorithm}\label{subsec:algorithm}
 We aim to estimate the true precision matrices $\left(\boldsymbol{\Omega}_1^*, \ldots, \boldsymbol{\Omega}_K^*\right)$ where $\boldsymbol{\Omega}_k^*=p_k^{-1/2}\bSigma_k^{*-1}$, $k=1, \ldots, K$. To address the identifiability issue in the parameterization of the tensor elliptical distribution, we assume that $\left\|\bOmega_k^*\right\|_F=$ 1 for $k=1, \ldots, K$. This renormalization does not change the graph structure of the original precision matrix.

A standard approach to estimating sparse precision matrices in tensor Gaussian graphical model is to minimize the penalized negative log-likelihood function \citep{sun+wang+liu+cheng-2015-Tlasso, xu2017gd, Min02012022},
 \begin{equation*}
     \tilde{q}_n(\bOmega_1,\ldots,\bOmega_K):=\frac{1}{p^*}\tr\{\hat\bSigma(\bOmega_K\otimes\cdots\otimes\bOmega_1)\}-\sum_{k=1}^K\frac{1}{p_k}\log\vert \bOmega_k\vert+\sum_{k=1}^K P_{\lambda_k}(\bOmega_k),
 \end{equation*}
 where $\hat{\bSigma}=n^{-1}\sum_{i=1}^n \vec(\T_i-\bmu)\vec(\T_i-\bmu)^\top$ 
 and $P_{\lambda_k}(\cdot)$ is a penalty function indexed by the tuning parameter $\lambda_k$. It builts on the sample covariance matrix and utilizes the framework of the famous Glasso method\citep{friedman+Hastie+Tibshirani-2008-glasso} originally developed for vector-valued data. However, the sample covariance matrix perform poorly when the data deviate from Gaussianity. 
 
 Spatial sign is an extension of the univariate sign to vectors. The spatial sign function and spatial-sign covariance matrix are defined as 
 \begin{equation*}
     U(\boldsymbol{x})=\|\boldsymbol{x}\|^{-1}\boldsymbol{x}\mathbb{I}(\boldsymbol{x}\neq \boldsymbol{0}) \text{~~and~~}\S=\E \left\{U(\boldsymbol x-\bmu)U(\boldsymbol x-\bmu)^\top\right\},
 \end{equation*}
  The spatial-sign covariance matrix retains the same eigenvectors of the sample covariance matrix\citep{oja2010multivariate}, and it is a robust and nonparametric alternatives to shape matrix estimation.

{Motivated by the robustness of spatial-sign covariance matrix, we define the spatial sign for tensor $\T$ as $U(\T)=\T/\Vert \T\Vert\mathbb{I}(\T\neq \boldsymbol{0})$, and we consider a spatial-sign based modified negative log-likelihood}
 \begin{equation*}
     l(\bOmega_1,\ldots,\bOmega_K):=\frac{1}{p^*}\tr\{\hat\S(\bOmega_K\otimes\cdots\otimes\bOmega_1)\}-\sum_{k=1}^K\frac{1}{p_k}\log\vert \bOmega_k\vert,
 \end{equation*}
where {$\hat\S=p^* n^{-1}\sum_{i=1}^n\vec\big(U(\T_i-{\bmu})\big)\vec\big(U(\T_i-{\bmu})\big)^\top$}. If $\bmu$ is unknown, we replace it with the sample spatial median $\hat{\bmu}$, that is,
\begin{equation}\label{eq:mu}
\hat{\bmu}=\argmin_{\bmu\in \T^{\boldsymbol p}} \sum_{i=1}^n \Vert \T_i-\bmu\Vert.
\end{equation}
 To encourage the sparsity of each precision matrix in the high dimensional scenario, we consider a penalized estimator which minimizes 
\begin{equation}\label{model:SSTGM}
    q_n(\bOmega_1,\ldots,\bOmega_K):=\frac{1}{p^*}\tr\{\hat\S(\bOmega_K\otimes\cdots\otimes\bOmega_1)\}-\sum_{k=1}^K\frac{1}{p_k}\log\vert \bOmega_k\vert+\sum_{k=1}^K P_{\lambda_k}(\bOmega_k),
\end{equation}
where $P_{\lambda_k}(\cdot)$ is a penalty function indexed by the tuning parameter $\lambda_k$ and we focus on the lasso penalty\citep{tibshirani-1996-lasso} $P_{\lambda_k}\left(\boldsymbol{\Omega}_k\right)=\lambda_k \sum_{i \neq j}\left|\left[\boldsymbol{\Omega}_k\right]_{i, j}\right|$. The estimation procedure is also suitable for a wide range of penalty functions, such as the SCAD penalty, the adaptive lasso penalty, the MCP penalty and the truncated $\ell_1$ penalty.

We call the model defined by Equation \eqref{model:SSTGM} as Sparse Spatial-sign Tensor Elliptical Graphical Model. The model  {is reduced to} the sparse tensor graphical model\citep{sun+wang+liu+cheng-2015-Tlasso} if the spatial-sign matrix $\hat\S$ 
is replaced by the vectorized sample covariance matrix $\hat{\bSigma}=n^{-1}\sum_{i=1}^n \vec(\T_i-\bmu)\vec(\T_i-\bmu)^\top$ and the data are assumed to follow a tensor normal distribution. Our framework is based on elliptical distributions, which include the tensor normal distribution as a special case and further allow for heavy-tailed distributions, and captures the graphical structure of higher-order tensor-valued data.

A natural approach to solving Equation \eqref{model:SSTGM} is to utilize a cyclic optimization scheme by leveraging the biconvexity of the objective function, as in \citet{sun+wang+liu+cheng-2015-Tlasso}, where  precision matrix for each mode is updated based on the estimates of the others. However,  {with} each update depending on the most recent estimates of all other modes,  {such approach} prevents parallelization and necessitates iterative procedures.  As a result, the computational cost can be substantial in high-dimensional settings. To address these issues, \citet{Min02012022} proposed a fast estimation procedure for the same  {optimization problem} that avoids iterative updates and  {demonstrates great performance} comparing with \cite{sun+wang+liu+cheng-2015-Tlasso}. It  {suggests} that directly solving the joint optimization problem may not be the  {most} efficient strategy,  {with the consideration of both computational efficiency and estimation accuracy}.

Motivated by \cite{Min02012022}, we propose an algorithm by estimating  {each $\bOmega_k$} separately. For estimating the $k$th precision matrix $ \bOmega_k $, we fix the remaining $ K-1 $ precision matrices $\{\tilde\bOmega_1,\ldots,\tilde\bOmega_{k-1},\tilde\bOmega_{k+1},\ldots,\tilde\bOmega_K\}$ and minimize the following loss function  {over $\bOmega_k$,} \begin{equation}\label{eq:bi-SSTGM}
    L(\bOmega_k):=\frac{1}{p_k}\tr(\hat\S_k\bOmega_k)-\frac{1}{p_k}\log \vert \bOmega_k\vert+\lambda_k\Vert\bOmega_k\Vert_{1,\text{off}},
\end{equation}
where $ \hat\S_k:=p_k n^{-1} \sum_{i=1}^n \V_i^k \V_i^{k \top}$, $\V_i^k:=\left[U(\T_i-{\bmu}) \times\left\{\tilde\bOmega_1^{1 / 2},\ldots, \tilde\bOmega_{k-1}^{1 / 2}, \I_{p_k}, \tilde\bOmega_{k+1}^{1 / 2}, \ldots, \tilde\bOmega_K^{1 / 2}\right\}\right]_{(k)}$ with $\times$ the tensor product operation and $[\cdot]_{(k)}$ the mode-$k$ matricization operation. \eqref{eq:bi-SSTGM}  {comes from the fact that} $\V_i^k=\left[U(\T_i-{\bmu})\right]_{(k)}\left(\tilde\bOmega_K^{1 / 2} \otimes \cdots \otimes \tilde\bOmega_{k+1}^{1 / 2} \otimes \tilde\bOmega_{k-1}^{1 / 2} \otimes \cdots \otimes \tilde\bOmega_1^{1 / 2}\right)^{\top}$, according to the properties of mode-$k$ matricization shown by \cite{kolda+Bader-2009-tensor}. Note that minimizing \eqref{eq:bi-SSTGM} corresponds to estimating vector-valued graphical model and can be solved efficiently via the glasso algorithm\citep{friedman+Hastie+Tibshirani-2008-glasso}. For any $\tilde{\bOmega}_\ell$ in $\V_i^k$, $i=1,\ldots,n,k=1,\ldots,K$, we restrict them  {such} that $\|\tilde{\bOmega}_\ell\|_2\lesssim p_\ell^{-1/2}$. This condition ensures that the eigenvalues of each initial precision matrix $\tilde{\bOmega}_\ell$ do not vary too widely.  {For example, one may choose} $\tilde{\bOmega}_\ell=p_\ell^{-1/2}\I_{\ell}$.  {We hereby adopt the same choice as in} \cite{Min02012022},  {i.e.}
\begin{equation}\label{eq:init_Omega}
    \begin{aligned}
        \tilde{\bOmega}_{\ell}=\begin{cases}
\left\{p_\ell n^{-1}\sum_{i=1}^n \left[U(\T_i-{\bmu})\right]_{(\ell)} \left[U(\T_i-{\bmu})\right]_{(\ell)}^\top\right\}^{-1}, & \quad np^* > p_\ell^2(p_\ell-1)/2; \\
\I_{p_\ell}, & \quad np^* \leq p_\ell^2(p_{\ell}-1)/2;
\end{cases}
    \end{aligned}
\end{equation}
 and normalized them to achieve $\Vert\tilde{\bOmega}_\ell\Vert_F=\Vert\hat{\bOmega}_\ell\Vert_F=1$, $\ell\in\{1,\ldots,K\}$. The details for solving precision matrix $\hat{\bOmega}_1,\ldots,\hat{\bOmega}_K$ are summarized in Algorithm \ref{alg:1}.
\begin{algorithm}[!ht]
    \renewcommand{\algorithmicrequire}{\textbf{Input:}}
	\renewcommand{\algorithmicensure}{\textbf{Output:}}
\caption{Solve sparse spatial-sign tensor elliptical graphical model via Spatial-Sign Separate tensor lasso}    \label{alg:1}
    \begin{algorithmic}[1] 
        \REQUIRE  Tensor $\T_1,\ldots,\T_n$, tuning parameters $\lambda_1,\ldots,\lambda_K$. 
	    \ENSURE $\hat\bOmega_1,\ldots,\hat\bOmega_K$; 
    \STATE Compute and normalize $\tilde{\bOmega}_{\ell}$, $\ell=1,\ldots,K$ by Equation \eqref{eq:init_Omega};
        \FOR {each $k=1,\ldots,K$}
            \STATE Solve Equation \eqref{eq:bi-SSTGM} for $\hat{\bOmega}_k$ via glasso \citep{friedman+Hastie+Tibshirani-2008-glasso};
            \STATE Normalize $\hat{\bOmega}_k$ such that $\Vert \hat{\bOmega}_k\Vert_F=1$.
        \ENDFOR
        \STATE \textbf{return} $\hat{\bOmega}_1^{(t)},\ldots,\hat{\bOmega}_K^{(t)}$.
    \end{algorithmic}
\end{algorithm}
It is clear that the estimation procedure of the precision matrices across different modes is mutually independent, thus our proposed method is separable. This separability allows for the use of parallel computing techniques to further accelerate the algorithm. 

\subsection{Theoretical Results}\label{subsec:theorem}

In this section, we show the estimation error for $\hat\bOmega_k$ and the proof is deferred to the appendix.
The assumptions are as follows:
\begin{assumption}\label{ass:eigenvalues}
    (Bounded eigenvalues) For any $k\in\{1,\ldots,K\}$, there is a constant $C_1>0$ such that,
    \begin{equation*}
        0<C_1\leq \lambda_{\min}(\bSigma_k^*)\leq \lambda_{\max}(\bSigma_k^*)\leq 1/C_1<\infty,
    \end{equation*}
    where $\lambda_{\min}(\bSigma_k^*)$ and $\lambda_{\max}(\bSigma_k^*)$ are the minimal and maximal eigenvalues of $\bSigma_k^*$ respectively.
\end{assumption}
\begin{assumption}\label{ass:lambda_inverse}
    There exists a constant $T>0$, $\max\{\Vert \bSigma^*\Vert_{L_1},\Vert (\bSigma^*)^{-1}\Vert_{L_1}\}\leq T$, for any $k\in \{1,\ldots,K\}$.
\end{assumption}
\begin{remark}
Assumption \ref{ass:eigenvalues} requires that the eigenvalues of the true covariance matrices are uniformly bounded. It implies that $\tr(\bSigma_k^*)\asymp p_k$ and $\Vert\bSigma_k^*\Vert_\infty<\infty$ for any $k\in\{1,\ldots,K\}$. Such an assumption is commonly used to establish the precision matrix estimation consistency in graphical models in the literatures\citep{he2014graphical,sun+wang+liu+cheng-2015-Tlasso,lyu2019tensor,Min02012022}. Assumption \ref{ass:lambda_inverse}  restricts the correlations across dimensions  {to be} not too large, which is also considered in \cite{lu2025robust}. Moreover, the constraint imposed on $(\bSigma^*)^{-1}$ corresponds to a special case of the matrix class introduced in \cite{bickel2008covariance,cai2011constrained}. 
\end{remark}
\begin{assumption}\label{ass:lambda}
    (Tuning) For any $k\in \{1,\ldots,K\}$ and some constant $C_2>0$, the tuning parameter $\lambda_k$ satisfies $1/C_2\{n^{-1/2}p_k^{1/2}(p^*)^{-1}(\log p_k)^{1/2}+ p_k^{-1}(p^*)^{-1/2}\}\leq \lambda_k\leq C_2\{n^{-1/2}p_k^{1/2}(p^*)^{-1}(\log p_k)^{1/2}+ p_k^{-1}(p^*)^{-1/2}\}$.
\end{assumption}

  Before characterizing the statistical error, we define a sparsity parameter for $\bOmega_k^*$, $k=1,\ldots,K$. Let $\mathbb S_k:=\{(i,j):[\bOmega_k^*]_{i,j}\neq 0\}$. Define the sparsity parameter as $s_k:=\vert \mathbb S_k\vert -p_k$, representing the number of nonzero entries in the off-diagonal component of $\bOmega_k^*$. 

\begin{theorem}\label{thm:estimation_error}
    Suppose the Assumption \ref{ass:eigenvalues}-\ref{ass:lambda} hold and $s_k=O(p_k)$ 
     for $k\in\{1,\ldots,K\}$. Then, for any $k\in \{1,\ldots,K\}$ 
    , we have,
    \begin{equation*}
      \Vert \hat{\bOmega}_k-\bOmega_k^*\Vert_F=O_p\left(\sqrt{\frac{(p_k+s_k)p_k\log p_k}{n p^*}}\right)+O\left(\sqrt{\frac{p_k+s_k}{p_k^2}}\right).
    \end{equation*}
\end{theorem}

Theorem \ref{thm:estimation_error} shows that the proposed estimator $\hat{\bOmega}_k$ converges to the true precision matrix $\bOmega_k^*$ at a rate of $\sqrt{(p_k+s_k)p_k\log p_k/(np^*)}+\sqrt{(p_k+s_k)/p_k^2}$ in  Frobenius norm. The first term of this rate has been established in several previous works\citep{he2014graphical,karahan2015tensor,Min02012022} and is minimax-optimal\citep{cai2016estimating}. The second term reflects the approximation error between the spatial-sign-based shape matrix and the true shape matrix, and becomes negligible as the dimension $p_k\rightarrow\infty$.

We next derive the estimation error in max norm and spectral norm. Similar with \cite{ravikumar2011high,lyu2019tensor}, some notations are introduced as follows. Denote $d_k$ as the maximum number of non-zeros in any row of the true precision matrices $\bOmega_k^*$, that is, $d_k=\max_{i}\vert\{j:[\bOmega_k^*]_{i,j}\neq 0\}\vert$, with $|\cdot|$ the set cardinality. 
Let the shape matrix be $\bLambda^*=p^*\bSigma/\tr(\bSigma)$ and $\bLambda_k^*=p_k\bSigma_k^*/\tr(\bSigma_k^*)$. For each shape matrix $\bLambda_k^*$, we define $\kappa_k^* := \|\bLambda_k^*\|_{L_\infty}$. Denote the Hessian matrix $\bGamma_k^* := \bOmega_k^{*-1} \otimes \bOmega_k^{*-1} \in \mathbb{R}^{p_k^2 \times p_k^2}$, whose entry $\bGamma_k^*[(i,j),(s,t)]$ corresponds to the second order partial derivative of the objective function with respect to $[\bOmega_k]_{i,j}$ and $[\bOmega_k]_{s,t}$. We define its sub-matrix indexed by $\mathbb S_k$ as $\bGamma_{k,\mathbb S_k,\mathbb S_k}^* = \left[\bOmega_k^{*-1} \otimes \bOmega_k^{*-1} \right]_{\mathbb S_k, \mathbb S_k}$, which is the $|\mathbb S_k| \times |\mathbb S_k|$ matrix with rows and columns of $\bGamma_k^*$ indexed by $\mathbb S_k$ and $\mathbb S_k$, respectively. Moreover, we define $\kappa_{\bGamma_k^*} := \| (\left[\bGamma_{k}^*\right]_{\mathbb S_k,\mathbb S_k})^{-1} \|_{L_\infty}$. In order to establish the rate of convergence in max norm, we need to impose an irrepresentability condition on the Hessian matrix.

\begin{assumption}[Irrepresentability]\label{ass:irrepresentability} For each $k = 1, \dots, K$, there exists some $\alpha_k \in (0,1]$ such that
\[
\max_{e \in \mathbb S_k^c} \left\| \left[ \bGamma_{k}^*\right]_{e,\mathbb S_k}\left( \left[\bGamma_{k}^*\right]_{\mathbb S_k,\mathbb S_k} \right)^{-1}  \right\|_1 \leq 1 - \alpha_k.
\]
\end{assumption}
\begin{remark}
    Assumption \ref{ass:irrepresentability} restricts the influence of the non-connected terms in $\mathbb S_k^c$ on the connected edges in $\mathbb S_k$. It has been widely used in developing the theoretical properties of lasso-type estimator \citep{zhao2006model,ravikumar2011high}, and is also considered in \cite{lyu2019tensor} for tensor precision matrices estimation.
\end{remark}

\begin{assumption}[Bounded Complexity]\label{ass:bounded_complexity}
For each $k = 1, \dots, K$, the parameters $\kappa_{\bSigma_k}^*$ and $\kappa_{\bGamma_k^*}$ are bounded, that is $\max\{\kappa_{\bSigma_k}^*,\kappa_{\bGamma_k^*}\}\leq C_3<\infty$, and the parameter $d_k$ satisfies $d_k = o \left( \left\{\sqrt{p_k \log p_k/(n p^*)}+1/p_k\right\}^{-1}\right)$.
\end{assumption}

\begin{theorem}\label{thm:estimation_error_spectral}
    Suppose the Assumption \ref{ass:eigenvalues}-\ref{ass:bounded_complexity} hold and $s_k=O(p_k)$ for $k\in\{1,\ldots,K\}$. Then, for any $k\in \{1,\ldots,K\}$, we have,
\begin{equation*}
        \Vert \hat{\bOmega}_k-\bOmega_k\Vert_\infty=O_p\left(\sqrt{\frac{p_k\log p_k}{n p^*}}\right)+O\left(\frac{1}{p_k}\right),
    \end{equation*}
    and
    \begin{equation*}
        \Vert \hat{\bOmega}_k-\bOmega_k\Vert_2=d_k\left\{O_p\left(\sqrt{\frac{p_k\log p_k}{n p^*}}\right)+O\left(\frac{1}{p_k}\right)\right\}.
    \end{equation*}
\end{theorem}
Theorem \ref{thm:estimation_error_spectral} establishes the rates of convergence under both max and spectral norm. If we ignore the approximation error, the first term of the rate under max norm achieves the optimal rate of convergence\citep{cai2016estimating}. Under the spectral norm, the first term of the rate matches that in the work of \cite{lyu2019tensor} and much faster than \cite{zhou2014gemini}. The approximation error has the rate $p_k^{-1}$ and $d_k p_k^{-1}$ and is negligible as $p_k\rightarrow\infty$.

Graphical models serve as powerful tools for capturing the dependence structure within tensor-valued data. In particular, the support of the conditional correlation matrix provides an effective characterization of linear dependencies. When the data follow an elliptical distribution, partial uncorrelatedness directly implies conditional uncorrelatedness \citep{baba2004partial}. Building on this insight, we adopt the approach of \cite{cai2011constrained} and introduce a new thresholded estimator $\breve{\bOmega}_k$ based on $\hat{\bOmega}_k$, with the entries 
\begin{equation*}
   [ \breve{\bOmega}_k]_{i,j}=[\hat{\bOmega}_k]_{i,j}\mathbb{I}\{[\hat{\bOmega}_k]_{i,j}\geq \tau_k\},
\end{equation*}
 where $\tau_k\geq 4 C_{k,\beta} (p^*)^{1/2}\lambda_k$ is a tuning parameter and $\lambda_k =C_{k,\beta} \{n^{-1/2}p_k^{1/2}(p^*)^{-1}(\log p_k)^{1/2}+ p_k^{-1}(p^*)^{-1/2}\}$, $\beta\in(0,1)$ and $C_{k,\beta}$ is a constant determined by $k$ and $\beta$. We define the following notation for support and sign patterns:
 \begin{equation*}
     \begin{aligned}
         \mathcal{M}(\breve{\bOmega}_k)=&\left\{\text{sign}([ \breve{\bOmega}_k]_{i,j}),1\leq i,j\leq p_k\right\},\\
          \mathcal{M}({\bOmega}_k^*)=&\left\{\text{sign}([ {\bOmega}_k^*]_{i,j}),1\leq i,j\leq p_k\right\},\\
          \mathcal{S}({\bOmega}_k^*)=&\left\{(i,j):[ {\bOmega}_k^*]_{i,j}\neq 0\right\},\\
          \theta_{\min}=&\min_{(i,j)\in\mathcal{S}({\bOmega}_k^*), 1\leq k\leq K}\vert[ {\bOmega}_k^*]_{i,j} \vert.
     \end{aligned}
 \end{equation*}
The condition for $\theta_{\min}$ is required to ensure that the nonzero entries are correctly retained. The threshold level $\tau_k$ serves as a tuning parameter, and its lower bound can be explicitly specified when $\lambda_k$, $n$, $p_k$, and $p^*$ are known. We have the following theorem, which is similar to the results of \cite{cai2011constrained} and \cite{lu2025robust}.

\begin{theorem}\label{thm:selection}
    Suppose the Assumption \ref{ass:eigenvalues}-\ref{ass:bounded_complexity} hold, $s_k=O(p_k)$ for $k\in\{1,\ldots,K\}$ and $\theta_{\min}>2\max_{1\leq k\leq K}\tau_k$. Then, for any $k\in \{1,\ldots,K\}$ and $\beta\in (0,1)$, with the probability larger than $1-\beta$, we have,
    \begin{equation*}
         \mathcal{M}(\breve{\bOmega}_k)= \mathcal{M}({\bOmega}_k^*).
    \end{equation*}
\end{theorem}

As shown in Theorem \ref{thm:selection}, the threshold estimator $\breve{\bOmega}_k$ not only recovers the sparsity pattern of $\bOmega_k^*$, but also identifies the signs of its nonzero entries, which are important properties commonly referred to as sign consistency.

\section{Simulation}
In this section, we compare three methods, named {\it Spatial-Sign Separate Estimator} (SSS), {\it Separate Estimator} (Sep), and {\it Cyclic Estimtor} (Cyc), respectively. Here, the Spatial-Sign Separate Estimator is our proposed method, the Separate Estimator is proposed in \cite{Min02012022} and the Cyclic Estimator is proposed in \cite{sun+wang+liu+cheng-2015-Tlasso}. First, for underlying distribution, we consider three different typical elliptical distributions,
\begin{itemize}
    \item \textbf{Tensor Normal Distribution:} $TN({\bm 0},{\bSigma})$;
    \item \textbf{Tensor $t_3$ Distribution:} $Tt({\bm 0}, {\bSigma},\nu =3)$;
    \item \textbf{Tensor Mixed Normal Distribution:} $TMN(0.2, 10, {\bm 0},{\bSigma})$.
\end{itemize}

Here, $TN({\bm 0},{\bLambda})$ denotes a tensor normal distribution with covariance matrix ${\bLambda}$, $Tt({\bm 0}, {\bLambda}, \nu)$ denotes a tensor t-distribution with degrees of freedom $\nu$ and covariance matrix ${\bLambda}$ and ${TMN}(\gamma, \sigma, {\bm 0}, {\bLambda})$ refers to a mixture tensor normal distribution $(1-\gamma) TN(0, {\bLambda}) + \gamma TN(0, \sigma^2 {\bLambda})$, where ${\bLambda}=\bLambda_1\otimes\cdots\otimes\bLambda_K$.

For covariance matrix structure, we follow similar settings in \cite{Min02012022} . We first introduce the following  covariance structures.
\begin{itemize}
    \item \textbf{Triangle (TR)} covariance. We set $\left[{\bf \Sigma}^*_k\right]_{i,j} = \exp\left(-|h_i - h_j|/2\right)$ with $h_1 < h_2 < \cdots < h_{p_k}$. The difference $h_i - h_{i-1}, i = 2, \ldots, p_k$, is generated i.i.d. from $\text{Unif}(0.5, 1)$. This generated covariance matrix mimics autoregressive process of order one, i.e., AR(1).
    
    \item \textbf{Autoregressive (AR)} precision. We set $\left[{\bf \Omega}^*_k\right]_{i,j} = 0.8^{|i-j|}$ and normalize the resulting matrix.
    
    \item \textbf{Compound symmetry (CS)} precision. We set $\left[{\bf \Omega}^*_k\right]_{i,j} = 0.6$ for $i \neq j$ and $\left[{\bf \Omega}^*_k\right]_{i,i} = 1$ and then normalize the resulting matrix.
\end{itemize}

Note that the first model (TR) owns a sparse precision matrix while the other two (AR and CS) have a non-sparse precision matrix.

We consider six models as follows. Models 1--3 are fully sparse models, and Models 4--6 are partially sparse models. In Models 3 \& 6, for better computational efficiency, we alternatively choose $(p_1,p_2,p_3)=(10,10,50)$ as the unbalanced model in terms of dimensions. In all models, we normalize the precision matrices such that $\|{\bf \Omega}^*_k\|_F = 1$ for $k = 1, \ldots, K$. For all the following models, we set $K = 3$ and $n = 100$.

\begin{itemize}
    \item \textbf{Model 1:} ${\bf\Omega}^*_1, {\bf \Omega}^*_2, {\bf \Omega}^*_3$ are all from the TR covariance model, $(p_1, p_2, p_3) = (30, 36, 30)$.
    
    \item \textbf{Model 2:} ${\bf\Omega}^*_1, {\bf \Omega}^*_2, {\bf \Omega}^*_3$ are all from the TR covariance model, $(p_1, p_2, p_3) = (100, 100, 100)$.
    
    \item \textbf{Model 3:} ${\bf\Omega}^*_1, {\bf \Omega}^*_2, {\bf \Omega}^*_3$ are all from the TR covariance model, $(p_1, p_2, p_3) = (10, 10, 50)$.
    
    \item \textbf{Model 4:} Same as Model 1, except for ${\bf \Omega}^*_1 = \text{AR}(0.8)$.
    
    \item \textbf{Model 5:} Same as Model 1, except for ${\bf \Omega}^*_1 = {\bf \Omega}^*_2 = \text{CS}(0.6)$.
    
    \item \textbf{Model 6:} Same as Model 3, except for ${\bf \Omega}^*_1 = {\bf \Omega}^*_2 = \text{CS}(0.6)$.
\end{itemize}
For tuning process, we use independent identically distributed validation samples to choose tuning parameter $\lambda$ for each method in each scenario. For methods other than ours, we  use the one with the smallest likelihood loss on the validation samples, where the likelihood loss is defined by
$$
L(\hat{\bf \Sigma}_k, \boldsymbol{\Omega})=\langle\boldsymbol{\Omega}, \boldsymbol{\Sigma}_n\rangle-\log \operatorname{det}(\boldsymbol{\Omega}),
$$
where $\hat {\bf \Sigma}_k :=p_k /{p^*n} \sum_{i=1}^n \V_i^k \V_i^{k \top}$ and $\V_i^k:=\left[(\T_i-{\bmu}) \times\left\{\tilde\bOmega_1^{1 / 2},\ldots, \tilde\bOmega_{k-1}^{1 / 2}, \I_{p_k}, \tilde\bOmega_{k+1}^{1 / 2}, \ldots, \tilde\bOmega_K^{1 / 2}\right\}\right]_{(k)}$ is the sample covariance matrix as defined in \cite{Min02012022}.  While for spatial-sign based methods, we use a similar likelihood loss:
$$
L(\hat{\bf S}_k, \boldsymbol{\Omega})=\langle\boldsymbol{\Omega}, {\hat {\bf S}}_k\rangle-\log \operatorname{det}(\boldsymbol{\Omega}),
$$
where $ \hat\S_k:=p_k n^{-1} \sum_{i=1}^n \V_i^k \V_i^{k \top}$ and $\V_i^k:=\left[U(\T_i-{\bmu}) \times\left\{\tilde\bOmega_1^{1 / 2},\ldots, \tilde\bOmega_{k-1}^{1 / 2}, \I_{p_k}, \tilde\bOmega_{k+1}^{1 / 2}, \ldots, \tilde\bOmega_K^{1 / 2}\right\}\right]_{(k)}$.

For fair comparison, we first normalize the traces of true covariance matrix and estimated covariance matrix to the corresponding dimensions. We then compute the loss in terms of Frobenius norm and Maximum norm for each mode. Besides, although all recovery rates highly depend on the scale of the tuning parameter, we still compute TPR and TNR for each mode to show the results of recovery as graphical models. All the results are based on 100 replications. 

All results are summarized in Table \ref{tab1}-\ref{tab6}. Here, we report the corresponding trace-normalized loss 
and the corresponding TNR for each mode, i.e. $k=1,2,3$ and also the average of them, denoted by AVG.  Based on all the tables, we can draw the following conclusions. First, in terms of recovery loss, the SSS estimator performs comparably to the better of the Sep and Cyc estimators under the multivariate Gaussian setting. In contrast, under heavy-tailed distributions, the SSS estimator consistently outperforms the others, demonstrating strong robustness with notably smaller standard errors across all replicates.

Second, regarding the true negative rate, the SSS estimator matches the best performance among the other two estimators in balanced models and in the mode with the highest dimension in unbalanced models under the multivariate Gaussian setting. However, in the heavy-tailed setting, the SSS estimator consistently outperforms both alternatives in these same configurations, again showing substantially better robustness through smaller standard errors across all scenarios.
\begin{table}[htbp]
\centering
\caption{\it Comparison of means and the standard errors (in parentheses) of different performance measures for Model 1 from 100 replicates. All methods have achieved 100\% true positive rate (and hence not shown in the table).}
{\tiny
\begin{tabular}{l *{3}{ccc}}
\toprule

  & \multicolumn{3}{c}{Multivariate Normal} 
  & \multicolumn{3}{c}{Multivariate $t_3$} 
  & \multicolumn{3}{c}{Multivariate Mixed Normal} \\
\cmidrule(lr){2-4} \cmidrule(lr){5-7} \cmidrule(lr){8-10}
      & SSS\ & Sep\ & Cyc\ 
      & SSS\ & Sep\ & Cyc\ 
      & SSS\ & Sep\ & Cyc\ \\
\midrule
\multicolumn{10}{c}{\bfseries Frobenius norm loss} \\
\addlinespace
AVG.  & 0.044(0.002)     & 0.044(0.003)     & 0.056(0.003)     &0.045(0.003)      &  0.336(0.233)    & 0.242(0.102)     & 0.045(0.003)     & 0.278(0.190)     & 0.211(0.021)     \\
$k=1$ & 0.041(0.004)      & 0.040(0.004)     &  0.049(0.004)    & 0.041(0.005)     &  0.312(0.219)    & 0.219(0.092)     & 0.041(0.004)     & 0.254(0.180)     & 0.201(0.021)     \\
$k=2$ & 0.049(0.004)      & 0.049(0.004)     &  0.067(0.004)    &  0.049(0.004)    & 0.360(0.241)     & 0.288(0.123)     & 0.050(0.004)     & 0.296(0.193)     & 0.263(0.026)     \\
$k=3$ &0.043(0.005)       & 0.043(0.005)     &   0.051(0.005)   &0.043(0.005)      &  0.337(0.246)    &   0.219(0.092)   &0.043(0.005)      & 0.283(0.208)     &   0.200(0.020)   \\
\midrule
\multicolumn{10}{c}{\bfseries Max norm loss} \\
\addlinespace
AVG.  & 0.011(0.001)     & 0.011(0.001)     &   0.011(0.001)   &0.011(0.001)      & 0.059(0.051)     &0.027(0.012)      &0.011(0.001)      &  0.051(0.044)    & 0.025(0.003)     \\
$k=1$ & 0.010(0.002)     & 0.010(0.002)     &   0.010(0.002)   &0.010(0.002)      & 0.055(0.048)    & 0.025(0.011)     &  0.010(0.002)    & 0.046(0.041)     &  0.023(0.004)    \\
$k=2$ & 0.011(0.002)     &  0.011(0.002)     & 0.011(0.002)     & 0.011(0.002)     & 0.055(0.044)     & 0.030(0.013)     &0.011(0.002)      &  0.047(0.039)    & 0.026(0.005)     \\
$k=3$ & 0.011(0.003)     &   0.011(0.003)    & 0.011(0.003)     &  0.011(0.003)    &  0.068(0.064)    & 0.027(0.012)     & 0.011(0.003)     & 0.060(0.055)      & 0.024(0.004)     \\
\midrule
\multicolumn{10}{c}{\bfseries True negative rate} \\
\addlinespace
AVG.  & 0.639(0.03)     &  0.663(0.03)    &   0.413(0.027)   &0.636(0.034)      &  0.434(0.402)    &  0(0)    & 0.636(0.036)     & 0.479(0.433)     & 0(0)     \\
$k=1$ & 0.632(0.07)     & 0.659(0.05)     & 0.438(0.033)     &0.620(0.075)      &  0.426(0.404)    & 0(0)     & 0.620(0.074)     & 0.473(0.433)     & 0(0)     \\
$k=2$ & 0.659(0.03)     &  0.679(0.03)    &  0.366(0.030)    & 0.657(0.033)     &  0.448(0.399)    & 0(0)     &  0.659(0.041)    &  0.493(0.435)    &  0(0)    \\
$k=3$ & 0.626(0.06)     &  0.650(0.05)    &  0.434(0.033)    & 0.630(0.064)     &   0.428(0.404)   &  0(0)    &  0.628(0.063)    &  0.472(0.434)    &  0(0)    \\
\bottomrule
\end{tabular}}
\label{tab1}
\end{table}

\begin{table}[htbp]
\centering
\caption{\it Comparison of means and the standard errors (in parentheses) of different performance measures for Model 2 from 100 replicates. All methods have achieved 100\% true positive rate (and hence not shown in the table). }
{\tiny
\begin{tabular}{l *{3}{ccc}}
\toprule
  & \multicolumn{3}{c}{Multivariate Normal} 
  & \multicolumn{3}{c}{Multivariate $t_3$} 
  & \multicolumn{3}{c}{Multivariate Mixed Normal} \\
\cmidrule(lr){2-4} \cmidrule(lr){5-7} \cmidrule(lr){8-10}
      & SSS\ & Sep\ & Cyc\ 
      & SSS\ & Sep\ & Cyc\ 
      & SSS\ & Sep\ & Cyc\ \\
\midrule
\multicolumn{10}{c}{\bfseries Frobenius norm loss} \\
\addlinespace
AVG.  & 0.027(0.001)     &  0.027(0.001)    &   0.046(0.001)   & 0.027(0.001)     &  0.157(0.147)   &   0.233(0.107)   & 0.027(0.001)     & 0.419(0.388)    &   0.220(0.023)    \\
$k=1$ & 0.027(0.001)     & 0.027(0.001)     &  0.046(0.001)    & 0.027(0.001)      & 0.157(0.147)     & 0.233(0.106)      &0.027(0.001)    & 0.411(0.379)    &  0.220(0.023)      \\
$k=2$ &  0.027(0.001)    & 0.027(0.001)     &   0.046(0.001)   & 0.027(0.001)      & 0.157(0.148)     & 0.233(0.106)     & 0.027(0.001)     & 0.428(0.397)    & 0.220(0.023)      \\
$k=3$ & 0.027(0.001)     & 0.027(0.001)     &    0.046(0.001)  &  0.027(0.001)     & 0.158(0.148)    &  0.233(0.107)    &  0.027(0.001)   & 0.418(0.397)    &   0.220(0.023)    \\
\midrule
\multicolumn{10}{c}{\bfseries Max norm loss} \\
\addlinespace
AVG.  & 0.004(0.000)     &  0.004(0.001)    &  0.004(0.000)    & 0.004(0.000)     &  0.010(0.005)   &  0.010(0.004)    &0.004(0.000)     &  0.046(0.051)   &  0.009(0.001)      \\
$k=1$ &0.004(0.001)      & 0.004(0.001)     &   0.004(0.001)   & 0.004(0.001)     &  0.010(0.005)     &   0.010(0.005)   & 0.004(0.001)     &  0.043(0.047)     &  0.009(0.002)         \\
$k=2$ & 0.004(0.001)     &  0.004(0.001)    &  0.004(0.001)    & 0.004(0.001)     &   0.009(0.005)   &   0.010(0.004)   & 0.004(0.001)     & 0.046(0.051)    &   0.009(0.001)        \\
$k=3$ & 0.004(0.001)     & 0.004(0.001)     &  0.004(0.001)    & 0.004(0.001)     &  0.010(0.005)   & 0.010(0.004)     & 0.004(0.001)     & 0.050(0.057)    & 0.009(0.001)           \\
\midrule
\multicolumn{10}{c}{\bfseries True negative rate} \\
\addlinespace
AVG.  & 0.774(0.007)     &  0.805(0.005)    &  0.406(0.004)    & 0.776(0.009)     & 0.393(0.355)     &  0(0)    &0.774(0.006)      & 0.577(0.465)    &   0(0)   \\
$k=1$ & 0.774(0.009)     &  0.806(0.009)    &  0.407(0.006)    & 0.773(0.005)     & 0.394(0.356)       &  0(0)    & 0.774(0.009)    &  0.577(0.466)   &  0(0)    \\
$k=2$ & 0.777(0.005)     &   0.803(0.009)   &   0.405(0.007)   &0.778(0.009)      & 0.391(0.354)        &  0(0)    &0.777(0.005)     & 0.577(0.465)    &  0(0)    \\
$k=3$ &  0.771(0.017)    &    0.806(0.006)  &   0.407(0.006)   &0.776(0.025)      & 0.393(0.357)        &  0(0)    & 0.771(0.017)    & 0.577(0.466)    &    0(0)  \\
\bottomrule
\end{tabular}}
\label{tab2}
\end{table}

\begin{table}[htbp]
\centering
\caption{\it Comparison of means and the standard errors (in parentheses) of different performance measures for Model 3 from 100 replicates. All methods have achieved 100\% true positive rate (and hence not shown in the table).}
{\tiny
\begin{tabular}{l *{3}{ccc}}
\toprule
  & \multicolumn{3}{c}{Multivariate Normal} 
  & \multicolumn{3}{c}{Multivariate $t_3$} 
  & \multicolumn{3}{c}{Multivariate Mixed Normal} \\
\cmidrule(lr){2-4} \cmidrule(lr){5-7} \cmidrule(lr){8-10}
      & SSS\ & Sep\ & Cyc\ 
      & SSS\ & Sep\ & Cyc\ 
      & SSS\ & Sep\ & Cyc\ \\
\midrule
\multicolumn{10}{c}{\bfseries Frobenius norm loss} \\
\addlinespace
AVG.  & 0.086(0.006)     &  0.083(0.006)    & 0.167(0.006)     & 0.088(0.005)     & 0.319(0.222)    & 0.494(0.258)      & 0.086(0.006)     &0.299(0.124)     & 0.447(0.042)      \\
$k=1$ &   0.038(0.007)   & 0.033(0.007)     &   0.034(0.007)   & 0.037(0.006)     & 0.150(0.060)    &   0.110(0.055)  &  0.038(0.007)    &0.154(0.100)     & 0.105(0.018)       \\
$k=2$ &  0.039(0.006)    & 0.034(0.005)     &  0.035(0.005)    &  0.038(0.006)    &  0.151(0.060)   & 0.111(0.058)     &    0.039(0.006)   & 0.156(0.104)    &  0.103(0.016)     \\
$k=3$ &  0.180(0.014)    &  0.181(0.013)    & 0.431(0.014)     & 0.188(0.012)     &  0.656(0.607)   & 1.262(0.665)     & 0.180(0.014)       & 0.587(0.209)    &   1.134(0.107)    \\
\midrule
\multicolumn{10}{c}{\bfseries Max norm loss} \\
\addlinespace
AVG.  &  0.021(0.003)    &   0.020(0.003)   & 0.021(0.003)     &  0.020(0.003)    & 0.062(0.024)    &  0.054(0.028)    & 0.021(0.003)     & 0.066(0.038)    &  0.048(0.007)      \\
$k=1$ & 0.012(0.004)     &  0.012(0.004)    & 0.012(0.004)     & 0.012(0.003)     &  0.047(0.021)   &  0.030(0.018)    &   0.012(0.004)    & 0.049(0.035)    &  0.028(0.007)       \\
$k=2$ & 0.012(0.003)     &  0.012(0.003)    &0.012(0.003)      & 0.012(0.003)     &  0.050(0.025)   & 0.031(0.018)     &   0.012(0.003)   &  0.053(0.042)   &  0.028(0.007)       \\
$k=3$ &  0.038(0.008)    &  0.037(0.008)    & 0.039(0.007)     &   0.036(0.007)   & 0.089(0.044)    & 0.100(0.052)     &  0.038(0.008)     & 0.097(0.039)    & 0.087(0.015)         \\
\midrule
\multicolumn{10}{c}{\bfseries True negative rate} \\
\addlinespace
AVG.  & 0.396(0.066)     &  0.536(0.059)    & 0.360(0.046)     &  0.357(0.060)     & 0.431(0.330)     &  0(0)    &  0.395(0.065)    & 0.485(0.322)    & 0(0)     \\
$k=1$ & 0.216(0.142)     & 0.441(0.117)     &  0.494(0.082)    & 0.204(0.128)     & 0.362(0.355)    &  0(0)    &  0.214(0.143)    & 0.396(0.346)    & 0(0)     \\
$k=2$ &0.215(0.138)      & 0.454(0.117)     & 0.502(0.077)     &  0.216(0.126)    & 0.373(0.361)    & 0(0)     & 0.215(0.142)      &  0.408(0.351)   &   0(0)   \\
$k=3$ &  0.757(0.011)    & 0.712(0.012)     & 0.084(0.011)     &  0.652(0.012)    &   0.558(0.306)   &  0(0)    & 0.755(0.021)   & 0.650(0.297)   & 0(0)     \\
\bottomrule
\end{tabular}}
\label{tab3}
\end{table}

\begin{table}[htbp]
\centering
\caption{\it Comparison of means and the standard errors (in parentheses) of different performance measures for Model 4 from 100 replicates. All methods have achieved 100\% true positive rate (and hence not shown in the table).}
{\tiny
\begin{tabular}{l *{3}{ccc}}
\toprule
  & \multicolumn{3}{c}{Multivariate Normal} 
  & \multicolumn{3}{c}{Multivariate $t_3$} 
  & \multicolumn{3}{c}{Multivariate Mixed Normal} \\
\cmidrule(lr){2-4} \cmidrule(lr){5-7} \cmidrule(lr){8-10}
      & SSS\ & Sep\ & Cyc\ 
      & SSS\ & Sep\ & Cyc\ 
      & SSS\ & Sep\ & Cyc\ \\
\midrule
\multicolumn{10}{c}{\bfseries Frobenius norm loss} \\
\addlinespace
AVG.  &0.068(0.006)      &0.063(0.005)      & 4.190(0.003)     &0.068(0.006)      &0.474(0.423)    & 4.322(0.073)     & 0.068(0.006)    &0.465(0.454)     &  4.308(0.014)     \\
$k=1$ &0.111(0.017)      &0.095(0.012)      &12.452(0.001)     &0.111(0.018)      &0.971(1.251)    &12.460(0.004)     & 0.110(0.017)    &0.817(1.004)     & 12.460(0.003)     \\
$k=2$ &0.049(0.005)      &0.049(0.005)      & 0.067(0.006)     &0.049(0.005)      &0.237(0.138)    & 0.288(0.123)     & 0.050(0.005)    &0.301(0.193)     &  0.263(0.026)     \\
$k=3$ &0.044(0.005)      &0.046(0.005)      & 0.051(0.005)     &0.044(0.005)      &0.214(0.102)    & 0.219(0.092)     & 0.044(0.005)    &0.276(0.186)     &  0.200(0.020)     \\
\midrule
\multicolumn{10}{c}{\bfseries Max norm loss} \\
\addlinespace
AVG.  &0.011(0.001)      &0.011(0.001)      & 0.497(0.002)     &0.011(0.001)      &0.046(0.031)    & 0.509(0.008)     & 0.011(0.001)    &0.053(0.048)     &  0.507(0.004)     \\
$k=1$ &0.011(0.001)      &0.010(0.001)      & 1.468(0.004)     &0.011(0.001)      &0.064(0.071)    & 1.470(0.011)     & 0.011(0.001)    &0.054(0.056)     &  1.470(0.009)     \\
$k=2$ &0.011(0.002)      &0.012(0.002)      & 0.011(0.002)     &0.011(0.002)      &0.034(0.013)    & 0.030(0.013)     & 0.011(0.002)    &0.047(0.039)     &  0.026(0.005)     \\
$k=3$ &0.011(0.003)      &0.011(0.003)      & 0.011(0.003)     &0.011(0.003)      &0.040(0.021)    & 0.027(0.012)     & 0.011(0.003)    &0.059(0.051)     &  0.024(0.004)     \\
\midrule
\multicolumn{10}{c}{\bfseries True negative rate} \\
\addlinespace
AVG.  &NA                &NA                & NA               &NA                &NA              & NA               &NA               &NA               & NA               \\
$k=1$ &NA                &NA                & NA               &NA                &NA              & NA               &NA               &NA               & NA               \\
$k=2$ &0.683(0.059)      &0.736(0.031)      & 0.370(0.050)     &0.679(0.065)      &0.438(0.397)    & 0(0)             &0.675(0.071)     &0.490(0.435)     & 0(0)             \\
$k=3$ &0.623(0.099)      &0.579(0.132)      & 0.438(0.057)     &0.626(0.098)      &0.415(0.410)    & 0(0)             &0.622(0.099)     &0.479(0.434)     & 0(0)             \\
\bottomrule
\end{tabular}}
\label{tab4}
\end{table}

\begin{table}[htbp]
\centering
\caption{\it Comparison of means and the standard errors (in parentheses) of different performance measures for Model 5 from 100 replicates. All methods have achieved 100\% true positive rate (and hence not shown in the table).}
{\tiny
\begin{tabular}{l *{3}{ccc}}
\toprule
  & \multicolumn{3}{c}{Multivariate Normal} 
  & \multicolumn{3}{c}{Multivariate $t_3$} 
  & \multicolumn{3}{c}{Multivariate Mixed Normal} \\
\cmidrule(lr){2-4} \cmidrule(lr){5-7} \cmidrule(lr){8-10}
      & SSS\ & Sep\ & Cyc\ 
      & SSS\ & Sep\ & Cyc\ 
      & SSS\ & Sep\ & Cyc\ \\
\midrule
\multicolumn{10}{c}{\bfseries Frobenius norm loss} \\
\addlinespace
AVG.  &0.122(0.015)     &0.081(0.008)     &13.879(0.002)     &0.122(0.015)     &0.841(1.050)    &13.938(0.032)    &0.122(0.015)     &0.738(0.880)     &13.932(0.007)  \\
$k=1$ &0.127(0.025)     &0.087(0.015)     &18.985(0.001)     &0.127(0.025)     &1.080(1.513)    &18.989(0.002)    &0.127(0.025)     &0.893(1.192)     &18.989(0.002)     \\
$k=2$ &0.194(0.034)     &0.114(0.019)     &22.602(0.001)     &0.194(0.034)     &1.230(1.634)    &22.606(0.003)    &0.194(0.034)     &1.045(1.302)     &22.606(0.002)   \\
$k=3$ &0.044(0.005)     &0.043(0.005)     &0.051(0.005)      &0.044(0.005)     &0.213(0.100)    &0.219(0.092)     &0.044(0.005)     &0.277(0.195)     &0.200(0.020)      \\
\midrule
\multicolumn{10}{c}{\bfseries Max norm loss} \\
\addlinespace
AVG.  &0.011(0.001)     &0.010(0.001)     &0.848(0.002)      &0.011(0.001)     &0.047(0.037)    &0.855(0.007)     &0.011(0.001)     &0.049(0.043)     &0.854(0.005)      \\
$k=1$ &0.011(0.002)     &0.009(0.001)     &1.268(0.004)      &0.011(0.002)     &0.051(0.051)    &1.270(0.011)     &0.011(0.002)     &0.043(0.042)     &1.270(0.009)       \\
$k=2$ &0.013(0.001)     &0.010(0.001)     &1.266(0.004)      &0.013(0.001)     &0.052(0.046)    &1.268(0.013)     &0.013(0.001)     &0.045(0.039)     &1.269(0.009)      \\
$k=3$ &0.011(0.003)     &0.011(0.003)     &0.011(0.003)      &0.011(0.003)     &0.040(0.021)    &0.027(0.012)     &0.011(0.003)     &0.059(0.053)     &0.024(0.004)        \\
\midrule
\multicolumn{10}{c}{\bfseries True negative rate} \\
\addlinespace
AVG.  &NA               &NA               &NA                &NA               &NA              &NA               &NA               &NA               &NA               \\
$k=1$ &NA               &NA               &NA                &NA               &NA              &NA               &NA               &NA               &NA               \\
$k=2$ &NA               &NA               &NA                &NA               &NA              &NA               &NA               &NA               &NA               \\
$k=3$ &0.623(0.098)     &0.652(0.094)     &0.438(0.057)      &0.625(0.097)     &0.414(0.424)           &0(0)             &0.622(0.098)     &0.480(0.434)     &0(0)            \\
\bottomrule
\end{tabular}}
\label{tab5}
\end{table}

\begin{table}[htbp]
\centering
\caption{\it Comparison of means and the standard errors (in parentheses) of different performance measures for Model 6 from 100 replicates. All methods have achieved 100\% true positive rate (and hence not shown in the table).}
{\tiny
\begin{tabular}{l *{3}{ccc}}
\toprule
  & \multicolumn{3}{c}{Multivariate Normal} 
  & \multicolumn{3}{c}{Multivariate $t_3$} 
  & \multicolumn{3}{c}{Multivariate Mixed Normal} \\
\cmidrule(lr){2-4} \cmidrule(lr){5-7} \cmidrule(lr){8-10}
      & SSS\ & Sep\ & Cyc\ 
      & SSS\ & Sep\ & Cyc\ 
      & SSS\ & Sep\ & Cyc\ \\
\midrule
\multicolumn{10}{c}{\bfseries Frobenius norm loss} \\
\addlinespace
AVG.  & 0.088(0.006)     & 0.089(0.006)     & 4.760(0.007)     &  0.088(0.007)    &0.561(0.403)     & 5.044(0.223)     & 0.088(0.007)     &0.413(0.266)     & 5.001(0.036)      \\
$k=1$ &0.042(0.009)      & 0.043(0.010)     &  6.940(0.001)    &  0.042(0.009)    &0.475(0.489)     &   6.948(0.005)   & 0.042(0.009)     & 0.339(0.350)    &       6.947(0.004) \\
$k=2$ & 0.042(0.008)     & 0.043(0.008)     &  6.915(0.001)    &   0.042(0.008)   &0.458(0.488)     &  6.922(0.005)    & 0.042(0.008)     & 0.319(0.325)    & 6.922(0.005)      \\
$k=3$ &  0.180(0.014)    & 0.181(0.013)     & 0.425(0.021)     &0.180(0.014)      &0.750(0.041)     &  1.262(0.665)    & 0.180(0.014)     & 0.580(0.216)    &   1.134(0.108)    \\
\midrule
\multicolumn{10}{c}{\bfseries Max norm loss} \\
\addlinespace
AVG.  & 0.019(0.003)     &  0.019(0.003)    &0.870(0.003)      & 0.019(0.003)     & 0.075(0.041)    & 0.894(0.018)     & 0.020(0.003)     & 0.065(0.037)    &   0.8895(0.0073)     \\
$k=1$ &  0.010(0.002)    &  0.010(0.002)    & 1.267(0.006)     &0.010(0.002)      & 0.065(0.052)    &   1.272(0.014)   & 0.010(0.002)     & 0.050(0.039)    &     1.272(0.013)    \\
$k=2$ & 0.010(0.002)     & 0.010(0.002)     &  1.303(0.005)    & 0.010(0.002)     &0.063(0.053)     & 1.309(0.017)     & 0.010(0.002)     & 0.048(0.035)    &   1.310(0.012)      \\
$k=3$ &  0.038(0.008)    & 0.037(0.008)     & 0.039(0.007)     &0.038(0.008)      &0.097(0.052)     & 0.100(0.052)     & 0.038(0.008)     &  0.097(0.039)   &   0.087(0.015)       \\
\midrule
\multicolumn{10}{c}{\bfseries True negative rate} \\
\addlinespace
AVG.  &  NA    &NA      & NA     &  NA    & NA    &  NA    &    NA  &  NA   &    NA  \\
$k=1$ & NA     & NA     & NA     &  NA    &    NA &  NA    &    NA  &  NA   &   NA   \\
$k=2$ &   NA   & NA     & NA     & NA     &  NA   &  NA    &    NA  &  NA   &  NA    \\
$k=3$ &0.757(0.012)      & 0.712(0.013)     & 0.093(0.022)     & 0.755(0.024)     & 0.538(0.322)     & 0(0)     & 0.755(0.023)   & 0.656(0.294)   &0(0)\\
\bottomrule
\end{tabular}}
\label{tab6}
\end{table}

\section{Real Data Application}
In this section, we apply the proposed method on the electroencephalography (EEG) data from a study to examine EEG correlates of genetic predisposition to alcoholism as in \cite{Min02012022} and \cite{li2010dimension}. The dataset is available at \url{http://kdd.ics.uci.edu/databases/eeg/}. There were 122 subjects among which 77 were alcohol individuals and 45 were nonalcoholic individuals. Each subject had 120 trials under exposure to different picture stimuli and the measurements from 64 electrodes are placed on subject's scalps which were sampled at 256 Hz (3.9-msec epoch) for 1 second.  Roughly speaking, the original data can be viewed as 122 tensor samples with dimension 256(sample time points)$\times$ 64(electrodes/channels) $\times$ 120(trials), which could be partitioned into two groups, the alcoholic one and the non-alcoholic one. Similar as \cite{li2010dimension}, we focus on  average results over trials under exposure to single stimulus. To save computational cost, we further downgrade the random matrix from 256 $\times$ 64 to 64$\times$ 64 as in \cite{Min02012022}. We divided the dataset into an alcoholic group and a nonalcoholic group. Before analyzing the estimated precision matrix, we first test the underlying data distribution of each group. 
\begin{figure}[htbp]
    \caption{\it QQ plots of different groups }
    \centering
    \begin{minipage}{0.4\textwidth}
        \includegraphics[width=\textwidth]{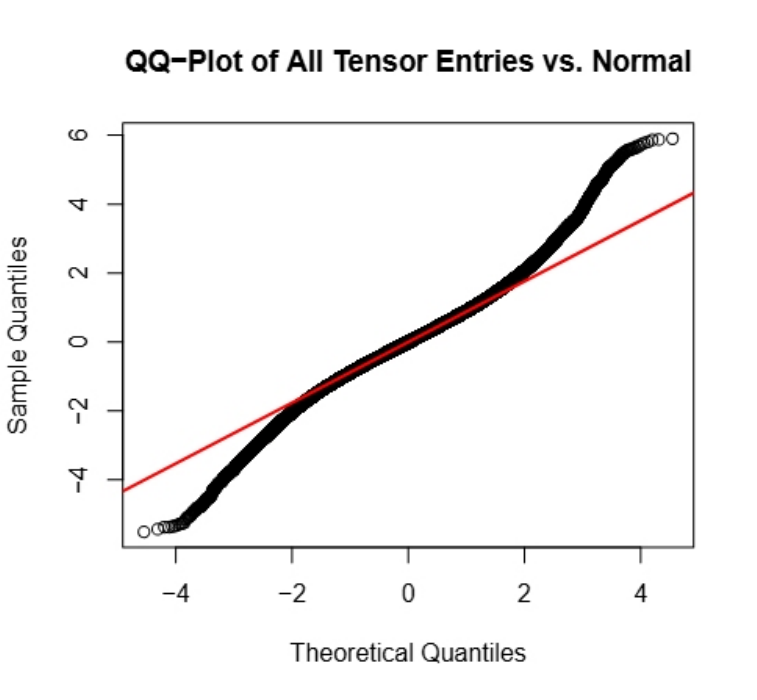}
        \parbox[t]{\textwidth}{\centering (a) Non-alcoholic group}
        \label{fig:noqq}
    \end{minipage}
    \begin{minipage}{0.4\textwidth}
        \includegraphics[width=\textwidth]{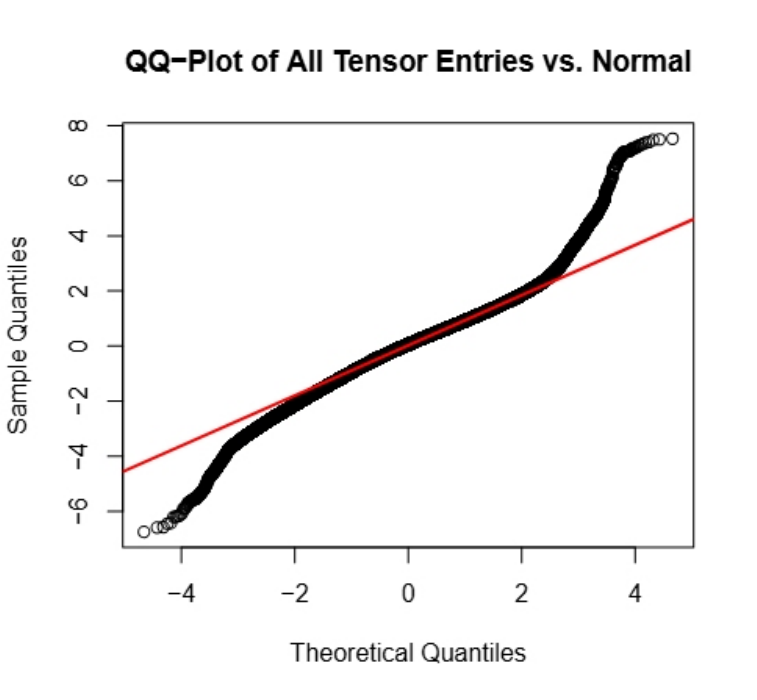}
       \parbox[t]{\textwidth}{\centering (b) Alcoholic group}
        \label{fig:alcoqq}
    \end{minipage}
    
    \label{fig:QQ}
\end{figure}

 Figure \ref{fig:QQ} illustrates the results, from which we can find that the underlying distributions of EEG data in both groups are far away from normal distribution and exhibit heavy-tailed characteristics. Therefore, it is more meaningful to use spatial-sign based methods to gain reliable and robust analysis of graphical models. To analyze the data, for each group, we standardized the data and applied the proposed method to estimate the precision matrix for channels. Tuning parameter was chosen by five-fold cross-validation.

 \begin{figure}[htbp]
    \caption{ \it Correlation networks constructed by the proposed method using the estimated precision matrices among channels for alcoholic group and nonalcoholic group. Only the first 100 strongest correlations are displayed. Nodes are labeled with EEG electrode identifiers. The blue nodes are placed on the middle of the scalp whose left and right nodes are placed on the left and right of the scalp respectively.  The thickness of edges represents the magnitude of correlations. }
    \centering
    \begin{minipage}{0.4\textwidth}
        \includegraphics[width=\textwidth]{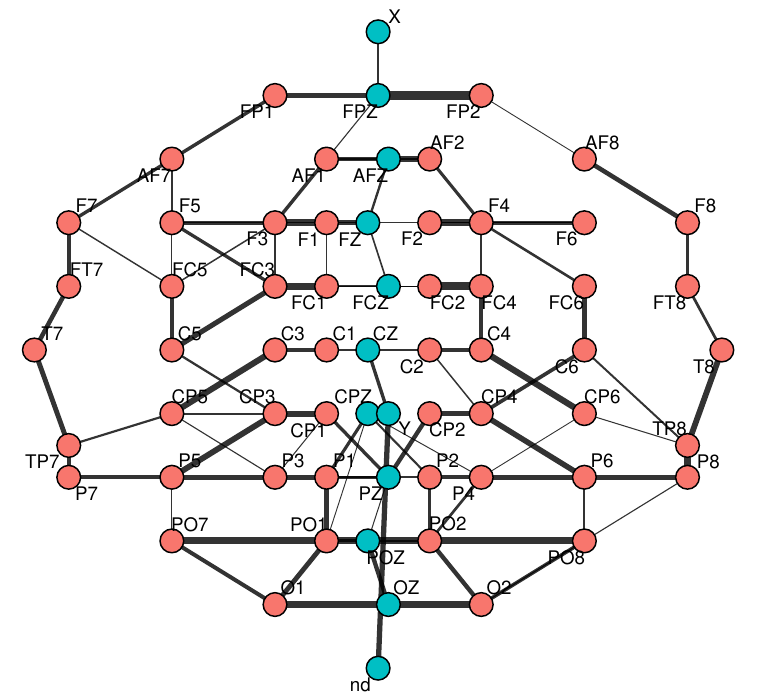}
        \parbox[t]{\textwidth}{\centering (a) Non-alcoholic group}
        \label{fig:no}
    \end{minipage}
    \begin{minipage}{0.4\textwidth}
        \includegraphics[width=\textwidth]{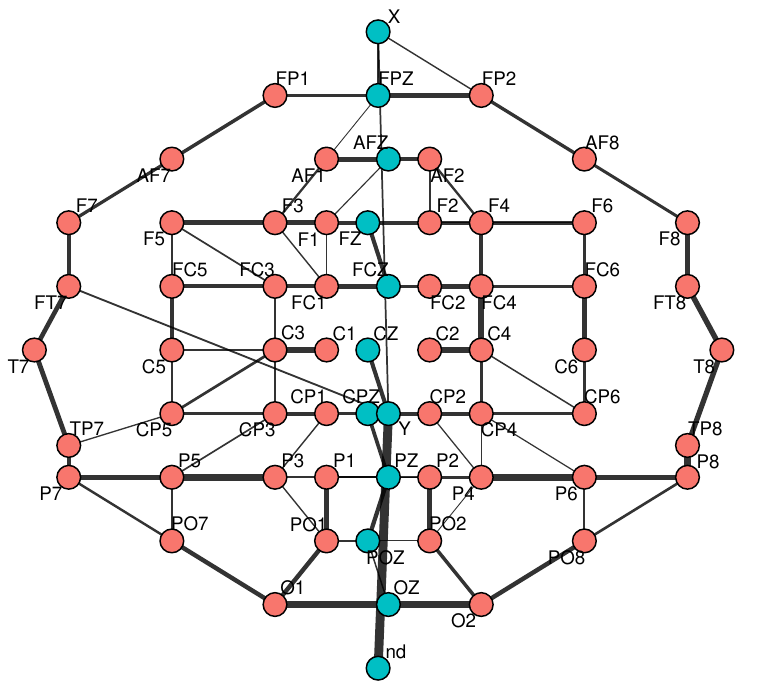}
       \parbox[t]{\textwidth}{\centering (b) Alcoholic group}
        \label{fig:alco}
    \end{minipage}
    
    \label{fig:cor}
\end{figure}

Figure \ref{fig:cor} illustrates the corresponding correlation networks. We follow the similar way of plotting as in \cite{Min02012022}. From the figure, some new different patterns that the analysis in \cite{Min02012022} does not exhibit can be recognized. For example, the connection pattern between 'F4', 'FC4', 'F6','FC6' is different between two groups. Besides, 'FC5' and 'FC3' correlates less strong in the  Non-alcoholic group, which is not detected in the old analysis. Overall, those differences may offer new insights for biologists to understand genetic predisposition to alcoholism better. 
\section{Conclusion}

In this paper, we proposed a fast and robust estimation procedure for the tensor precision matrix by introducing a novel tensor elliptical graphical model built upon the spatial-sign methodology. The proposed estimator is computationally efficient and exhibits strong robustness properties, particularly under heavy-tailed distributions. We rigorously established its theoretical properties, including convergence rates, which demonstrate the reliability of the method in high-dimensional tensor settings.

Beyond precision matrix estimation, our framework holds great potential for broader applications. For example, it can be effectively integrated into downstream tasks such as tensor discriminant analysis \citep{min2023optimality,wang2024parsimonious}, tensor regression \citep{li2017parsimonious,wang2024robust} and tensor classification \citep{pan2019covariate}. These directions are particularly relevant in modern applications involving neuroimaging, genomics, and signal processing, where data naturally exhibit tensor structures and heavy-tailed behaviors.

\section{Appendix}

In appendix, we provide the proofs of all the theorems presented in the paper, along with the lemmas required for their proofs. 


\subsection{Proof of Theorem \ref{thm:estimation_error}}

Based on the Equation \eqref{model:SSTGM}, we define the population function as
\begin{equation}\label{eq:Population}
    \begin{aligned}
        &q(\bOmega_1,\ldots,\bOmega_K)\\
        &~~:=\frac{1}{p^{*}}\E \left[ \tr\{p^*/\tr(\bSigma)\vec(\T_i-\bmu)\vec(\T_i-\bmu)^\top(\bOmega_K\otimes\cdots\otimes\bOmega_1)\} \right]-\sum_{k=1}^K\frac{1}{p_k}\log\vert \bOmega_k\vert.
    \end{aligned}
\end{equation}

By minimizing the Equation \eqref{eq:Population} with respect to $\bOmega_1,\ldots\bOmega_K$ respectively, we obtain the population results for $\bOmega_1,\ldots,\bOmega_K$, denoting as 
\begin{equation*}
    \begin{aligned}
        M_k:=\argmin_{\bOmega_k} q(\bOmega_1,\ldots,\bOmega_K).
    \end{aligned}
\end{equation*}

From the Theorem 3.1 in \cite{sun+wang+liu+cheng-2015-Tlasso}, we have the explicit form 
\begin{equation}\label{eq:Mk}
    M_k=\frac{p^*}{p_k\prod_{j\neq k}\tr(\bLambda_j^*\bOmega_j)}(\bLambda_k^*)^{-1},
\end{equation}
where $\bLambda_k^*$ denotes the shape matrix of the $k$-th mode, $\bLambda_k^*=p_k\bSigma_k^*/\tr(\bSigma_k^*)$.

We first consider the statistical error for the sample minimization function $\hat{M}_k$, where 
\begin{equation*}
    \begin{aligned}
        &\hat{M}_{k}:=\hat{M}_{k}(\tilde\bOmega_1,\ldots,\tilde\bOmega_{k-1},\tilde\bOmega_{k+1},\ldots,\tilde\bOmega_K)\\
        =&\argmin_{\bOmega_k} q_n(\tilde\bOmega_1,\ldots,\tilde\bOmega_{k-1},\bOmega_{k},\tilde\bOmega_{k+1},\ldots,\tilde\bOmega_K)=\argmin_{\bOmega_k} L(\bOmega_k).
    \end{aligned}
\end{equation*}

For some constant $H,C>0$, we define the set of convergence
\begin{equation*}
    \mathbb{A}:=\left\{ \mathbf{\Delta}\in \mathbb{R}^{p_k\times p_k}:\mathbf{\Delta}=\mathbf{\Delta}^\top,
    \|\mathbf{\Delta}\|_F=(p^*/p_k)\left\{H\sqrt{\frac{(p_k+s_k)p_k^3\log p_k}{n (p^*)^2}}+C\sqrt{\frac{p_k+s_k}{p^*}}\right\}\right\}.
\end{equation*}

The key idea is to show that for any $k\in\{1,\ldots,K\}$,
\begin{equation}\label{eq:thm1_Mk}
    \inf_{\mathbf{\Delta}\in \mathbb{A}}\left\{L(M_k+\mathbf{\Delta})-L(M_k)\right\}>0,
\end{equation}
with high probability. Note that the function $L(M_k+\mathbf{\Delta})-L(M_k)$ is convex in $\mathbf{\Delta}$ and since $\hat{M}_k$ minimizes $L(\bOmega_1)$, we have $L(\hat{M}_k)-L(M_k)\leq 0$. If we show Equation \eqref{eq:thm1_Mk}, then the minimizer $\hat{\mathbf{\Delta}}=\hat{M}_k-M_k$ must be within the interior of the ball define by $\mathbb{A}$ and hence $\|\hat{\mathbf{\Delta}}\|_F\leq (p^*/p_k)\left\{H\sqrt{(p_k+s_k)p_k^3\log p_k/\{n (p^*)^2\}}+C\sqrt{(p_k+s_k)/p^*}\right\}$. Similar technique is applied in \cite{fan2009SCAD,sun+wang+liu+cheng-2015-Tlasso,lyu2019tensor}. 

To show Equation \eqref{eq:thm1_Mk}, we decompose $\inf_{\mathbf{\Delta}\in \mathbb{A}}\left\{L(M_k+\mathbf{\Delta})-L(M_k)\right\}$ as three parts $I_1$, $I_2$ and $I_3$ with 
\begin{equation*}
    \begin{aligned}
        I_1&:=\frac{1}{p_k} \tr(\mathbf\Delta\hat{\S}_k)-\frac{1}{p_k}\left\{ \log \vert M_k+\mathbf\Delta\vert -\log \vert M_k\vert \right\},\\
        I_2&:=\lambda_k\left\{ \| [M_k+\mathbf\Delta]_{\mathbb S_k}\|_{1,\text{off}}-\| [M_k]_{\mathbb S_k}\|_{1,\text{off}} \right\},\\
        I_3&:=\lambda_k\left\{ \| [M_k+\mathbf\Delta]_{\mathbb S_k^c}\|_{1}-\| [M_k]_{\mathbb S_k^c}\|_{1} \right\}.\\
    \end{aligned}
\end{equation*}
It suffice to show that $I_1+I_2+I_3>0$ with high probability. To simplify the term $I_1$, we employ the Taylor expansion of $f(t)=\log \vert M_k+t\mathbf\Delta\vert$ at $t=0$ and obtain
\begin{equation*}
    \begin{aligned}
        \log\vert M_k+\mathbf\Delta\vert-\log\vert\mathbf\Delta\vert=&\tr(M_k^{-1}\mathbf\Delta)-\{\vec(\Delta)\}^\top \left\{ \int_0^1 (1-\nu)\mathbf M_\nu^{-1}\otimes \mathbf M_\nu^{-1}\text{d}\nu\right\}\vec(\Delta).
    \end{aligned}
\end{equation*}
where $\mathbf M_\nu:=M_k+\nu\mathbf\Delta\in \mathbb{R}^{p_k\times p_k}$. Consequently , We compose $I_1$ as
\begin{equation*}
    \begin{aligned}
        I_1=& \frac{1}{p_k}\tr\{(\hat{\S}_k-M_k^{-1})\mathbf\Delta \}+\frac{1}{p_k}\{\vec(\Delta)\}^\top \left\{ \int_0^1 (1-\nu)\mathbf M_\nu^{-1}\otimes \mathbf M_\nu^{-1}\text{d}\nu\right\}\vec(\Delta)\\
        :=&I_{11}+I_{12}.
    \end{aligned}
\end{equation*}
and $I_{11}$
\begin{equation*}
\begin{aligned}
    I_{11}\leq &\frac{1}{p_k}\left\vert \sum_{(i,j)\in\mathbb S_1}(\hat{S}_k-M_k^{-1})_{i,j}\mathbf\Delta_{i,j} \right\vert+\frac{1}{p_k}\left\vert \sum_{(i,j)\in\mathbb S_1}(\hat{S}_k^c-M_k^{-1})_{i,j}\mathbf\Delta_{i,j} \right\vert\\
    :=&I_{111}+I_{112}.
    \end{aligned}
\end{equation*}

Notice that, for two matrices $\A$, $\mathbf B$ and an index set $\mathbb S$, we have 
\begin{equation*}
    \left\vert \sum_{(i,j)\in\mathbb S}[\A]_{i,j}[\mathbf B]_{i,j}\right\vert\leq \max_{i,j}\left\vert [\A]_{i,j}\right\vert \left\vert \sum_{(i,j)\in\mathbb S}[\mathbf B]_{i,j}\right\vert \leq \sqrt{\vert \mathbb S\vert} \max_{i,j}\left\vert [\A]_{i,j}\right\vert \|\mathbf B\|_F,
\end{equation*}

For $I_{111}$, the condition for $\bOmega_\ell,\ell\neq k$ in Lemma \ref{lemma:Sk} is obviously holds.
By Lemma \ref{lemma:Sk},
\begin{equation*}
    \begin{aligned}
        I_{111}\leq& \frac{\sqrt{p_k+s_k}}{p_k}\max_{i,j}\left\vert[\hat{\S}_k-M_k^{-1}]_{i,j}\right\vert \|\mathbf\Delta\|_F\\
        \lesssim & \frac{\sqrt{p_k+s_k}}{p_k}(\sqrt{p_k^3\log p_k/\{n (p^*)^2\}}+\tilde{C}/\sqrt{p^*})\|\mathbf\Delta\|_F\\
        =&(p^*/p_k)p_k^{-1}(\sqrt{ (p_k+s_k)p_k^3\log p_k/\{n (p^*)^2\}}+\tilde{C}\sqrt{(p_k+s_k)/p^*})\cdot\\
        &~~~~~~~~~~~~~~~~~~(H\sqrt{ (p_k+s_k)p_k^3\log p_k/\{n (p^*)^2\}}+C\sqrt{(p_k+s_k)/p^*}),
    \end{aligned}
\end{equation*}
with high probability.

For $I_{12}$, by the property of the minimum eigenvalue, we have
\begin{equation*}
\begin{aligned}
        I_{12}\geq &\frac{1}{p_k}\|\vec(\mathbf\Delta)\|_2^2\int_0^1 (1-\nu)\lambda_{\min}(\mathbf M_\nu^{-1}\times\mathbf M_\nu^{-1})\text{d}\nu\\
        =&\frac{1}{p_k}\|\vec(\mathbf\Delta)\|_2^2\int_0^1 (1-\nu)\{\lambda_{\max}(M_k)+\nu\lambda_{\max}(\mathbf\Delta)\}^{-2}\text{d}\nu\\
        \geq &\frac{1}{2 p_k}\|\vec(\mathbf\Delta)\|_2^2\min_{0\leq \nu\leq 1}\{\lambda_{\max}(M_k)+\nu\lambda_{\max}(\mathbf\Delta)\}^{-2}\\
                \geq &\frac{1}{2 p_k}\|\vec(\mathbf\Delta)\|_2^2\{\|M_k\|_2+\|\mathbf\Delta\|_2\}^{-2}\\
\end{aligned}
\end{equation*}

By Assumption~\ref{ass:eigenvalues} and Equation~\eqref{eq:Mk}, $\|M_k\|_2\asymp (p^*/p_k)^{1/2}$ and $\|\mathbf\Delta\|_2$ satisfies $\|\mathbf\Delta\|_2\leq \|\mathbf\Delta\|_F=(p^*/p_k) \left\{H\sqrt{(p_k+s_k)p_k^3\log p_k/\{n (p^*)^2\}}+C\sqrt{(p_k+s_k)/p^*}\right\}$. We have,
\begin{equation*}
    I_{12}\gtrsim \frac{1}{p_k(p^*/p_k)} \|\vec(\mathbf\Delta)\|_2^2=(p^*/p_k)p_k^{-1}(H\sqrt{ (p_k+s_k)p_k^3\log p_k/\{n (p^*)^2\}}+C\sqrt{(p_k+s_k)/p^*})^2,
\end{equation*}
which dominates the term $I_{111}$ for sufficiently large $H$.

To bound $I_2$, we have,
\begin{equation*}
    \vert I_2\vert \leq \lambda_1 \| [\mathbf{\Delta}]_{\mathbb S_1}\|_{1,\text{off}}\leq \lambda_1 \sqrt{p_k+s_k}\|\mathbf{\Delta}\|_F.
\end{equation*}

By Assumption \ref{ass:lambda}, $\lambda_1\lesssim \{n^{-1/2}p_k^{1/2}(p^*)^{-1}(\log p_k)^{1/2}+ p_k^{-1}(p^*)^{-1/2}\}$. Therefore,
\begin{equation*}
    \begin{aligned}
       \vert I_2\vert\lesssim &(p^*/p_k)p_k^{-1}(\sqrt{ (p_k+s_k)p_k^3\log p_k/\{n (p^*)^2\}}+\sqrt{(p_k+s_k)/p^*})\cdot\\
        &~~~~~~~~~~~~~~~~~~(H\sqrt{ (p_k+s_k)p_k^3\log p_k/\{n (p^*)^2\}}+C\sqrt{(p_k+s_k)/p^*}),
    \end{aligned}
\end{equation*}
which is dominated by the term $I_{12}$ for sufficiently large $H$.

We next show that $I_3-\vert I_{112}\vert>0$ with high probability. For $I_3$, we have
\begin{equation*}
    I_3=\lambda_1\sum_{(i,j)\in\mathbb S_1^c}\left\{ \vert [M_k]_{i,j}+[\mathbf\Delta]_{i,j} \vert-\vert [\mathbf\Delta]_{i,j} \vert\right\}=\lambda_1\sum_{(i,j)\in\mathbb S_1^c} \vert [\mathbf\Delta]_{i,j} \vert 
\end{equation*}

Together with the expression of $I_{112}$ and the bound in Lemma \ref{lemma:Sk}, we have
\begin{equation*}
    \begin{aligned}
        I_{3}-I_{112}=\sum_{(i,j)\in\mathbb S_1^c}\left\{ 
\lambda_1- p_k^{-1}\left(\breve{C}\sqrt{\frac{p_k^3\log p_k}{n (p^*)^2}}+\frac{\tilde{C}}{\sqrt{p^*}}\right)\right\}\sum_{(i,j)\in\mathbb S_1^c} \vert [\mathbf\Delta]_{i,j} \vert >0,
    \end{aligned}
\end{equation*}
as long as $1/C_1>\min(\breve{C},\tilde{C})$ for some constant $\breve{C}$ and $\tilde{C}$,which is valid for sufficient small $C_1$ in Assumption \ref{ass:lambda}. Combing all these bounds together, we have, for any $\mathbf\Delta\in \mathbb A$, with high probability,
\begin{equation*}
    L(M_k+\mathbf\Delta)-L(M_k)\geq I_2-I_{111}-\vert I_2\vert+I_3-I_{112}>0.
\end{equation*}

The proof for $\hat{M}_k$ is completed, that is,
\begin{equation*}
    \|\hat{M}_k-M_k\|_F=p^*/p_k\left\{O_p\left(\sqrt{\frac{(p_k+s_k)p_k^3\log p_k}{n (p^*)^2}}\right)+O\left(\sqrt{\frac{p_k+s_k}{p^*}}\right)\right\}.
\end{equation*}

Then, for $\hat{\bOmega}_k$, by triangle inequality,
\begin{equation*}
    \begin{aligned}
        \|\hat{\bOmega}_k-\bOmega_k^*\|_F=&\left\|\frac{\hat{M}_k}{\|\hat{M}_k\|_F}-\frac{M_k}{\|M_k\|_F}\right\|_F\\
        \leq&\left\|\frac{\hat{M}_k}{\|\hat{M}_k\|_F}-\frac{M_k}{\|\hat{M}_k\|_F}\right\|_F+\left\|\frac{M_k}{\|\hat{M}_k\|_F}-\frac{M_k}{\|M_k\|_F}\right\|_F\\
        \leq &\frac{2}{\|\hat{M}_k\|_F}\left\| \hat{M}_k-M_k\right\|_F.
    \end{aligned}
\end{equation*}

By Assumption~\ref{ass:eigenvalues} and Equation~\eqref{eq:Mk}, $\|M_k\|_F\asymp (p^*)^{1/2}$. Then we bound $\hat{M}_k$ as 
\begin{equation*}
\begin{aligned}
    \|\hat{M}_k\|_F\gtrsim &\|M_k\|_F-\|\hat{M}_k-M_k\|_F\\
    =&(p^*)^{1/2}-(p^*)^{1/2}\left\{O_p\left(\sqrt{\frac{(p_k+s_k)p_k\log p_k}{n p^*}}\right)+O\left(\sqrt{\frac{p_k+s_k}{p_k^2}}\right)\right\}=(p^{*})^{1/2}(1+o_p(1)).
    \end{aligned}
\end{equation*}

Hence,
\begin{equation*}
    \|\hat{\bOmega}_k-\bOmega_k^*\|_F\leq O_p\left(\sqrt{\frac{(p_k+s_k)p_k\log p_k}{n p^*}}\right)+O\left(\sqrt{\frac{p_k+s_k}{p_k^2}}\right).
\end{equation*}

The proof is completed.

\subsection{Proof of Theorem \ref{thm:estimation_error_spectral}}

Before the proof, we introduce some  useful notations. 
Recall that
\begin{equation*}
    M_k:=\argmin_{\bOmega_k} q(\bOmega_1,\ldots,\bOmega_K)=\frac{p^*}{p_k\prod_{j\neq k}\tr(\bLambda_j^*\bOmega_j)}(\bLambda_k^*)^{-1}.
\end{equation*}
and
\begin{equation*}
      L(\bOmega_k):=\frac{1}{p_k}\tr(\hat\S_k\bOmega_k)-\frac{1}{p_k}\log \vert \bOmega_k\vert+\lambda_k\Vert\bOmega_k\Vert_{1,\text{off}}.
\end{equation*}
Let
\begin{equation*}
    \breve{M}_k=\argmin_{\bOmega_k\succ 0,[\bOmega_k]_{\mathbb S^c_k}=0} L(\bOmega_k),   
\end{equation*}
 $\bDelta_k=\breve{M}_k-M_k$, $\W_k= \S_k-\frac{p_k\prod_{j\neq k}\tr(\bLambda_j^*\bOmega_j)}{p^*}\bLambda_k^*$ and $R_k(\bDelta_k)=(M_k+\bDelta_k)^{-1}-M_k^{-1}+M_k^{-1}\bDelta_k M_k^{-1}$. We first show the convergence rate of $\breve{M}_k-M_k$.

    According to the matrix differentiation rule, \[
    \nabla^2_{\bOmega}(\tr(\hat{\S}_k {\bOmega})-\log\vert \bOmega\vert) = {\bOmega}^{-1}\otimes {\bOmega}^{-1}\succ 0,
    \]
   thus we have,
    \[
        \nabla^2_{{[\bOmega_k]}_{\mathbb S_k}}(\tr(\hat {\S}_k {\bOmega_k})-\log\vert \bOmega_k\vert) = [{\bOmega_k}^{-1}\otimes {\bOmega_k}^{-1}]_{\mathbb S_k \mathbb S_k}\succ 0
    \]
    By Lagrangian duality, we know $\breve{M}_k=\argmin _{{\bOmega_k}\succ 0, {[\bOmega_k]}_{\mathbb S^c}=0, \|{\bOmega_k}\|_1\leq C(\lambda_n)}\{\tr(\hat{\S}_k {\bOmega_k})-\log\vert \bOmega_k\vert\}$. Hence the strict convexity of the objective function in ${[\bOmega_k]}_{\mathbb S_k}$ implies the uniqueness of $\breve{M}_k$. Besides, by Karush-Kuhn-Tucker (KKT) conditions, we know that 
    \[
        {[\bOmega_k]}_{\mathbb S_k} = [\breve{M}_k]_{\mathbb S_k} \Longleftrightarrow G({[\bOmega_k]}_{\mathbb S_k}):=[\hat {\S}_k]_{\mathbb S_k}-{[\bOmega_k]}_{\mathbb S_k}^{-1}+p_k\lambda_k {[{\bf Z}_k]}_{\mathbb S_k}=0
    \]
    where ${[{\bf Z}_k]}$ belongs to the sub-differential of $\|\bOmega_k\|_{1,\text{off}}$, that is,
\begin{equation}\label{eq:Z_k}
[\mathbf{Z}_k]_{i,j} := 
\begin{cases}
0, & \text{if } i = j, \\
\operatorname{sign}([\bOmega_k]_{i,j}), & \text{if } i \neq j \text{ and } [\bOmega_k]_{i,j} \neq 0, \\
\in [-1, +1], & \text{if } i \neq j \text{ and } [\bOmega_k]_{i,j} = 0.
\end{cases}
\end{equation}

    Define another map $F:\mathbb R^{\vert\mathbb S_k\vert}\rightarrow\mathbb R^{\vert\mathbb S_k\vert}$, such that $F(\Vec({[\bDelta_k^\prime]}_{\mathbb S_k})):=-(p^*/p_k)({\bGamma}^*_{k,\mathbb S_k,\mathbb S_k})^{-1}{\Vec}(G([M_k]_{\mathbb S_k}+{[\bDelta_k^\prime]}_{\mathbb S_k}))+{\Vec}({[\bDelta_k^\prime]}_{\mathbb S_k})$. By definition, we know that $F(\Vec({[\bDelta_k^\prime]}_{\mathbb S_k}))=\Vec({[\bDelta_k^\prime]}_{\mathbb S_k})$ if and only if $\Vec(G([M_k]_{\mathbb S_k}+{[\bDelta_k^\prime]}_{\mathbb S_k}))=0$, i.e. ${[\bDelta_k^\prime]}_{\mathbb S_k}={[\bDelta_k]}_{\mathbb S_k}$.

   Denote $2C_3(p^*/p_k)(\|{\W_k}\|_\infty+p_k\lambda_k)$ by $r$ and $\{\Vec({[\bDelta_k^\prime]}_{\mathbb S_k}): \|{[\bDelta_k^\prime]}_{\mathbb S_k}\|_\infty \leq r\}$ by $B(r)$, we now \textbf{claim} that 
    \begin{equation}\label{eq:FBr}
        F(B(r))\subseteq B(r).
    \end{equation}
    Since $B(r)$ is a nonempty compact convex set, then by Brouwer fixed-point theorem, we can know that $\Vec({[\bDelta_k]}_{\mathbb S_k}) \in B(r)$, and hence
    \begin{equation}\label{eq:Delta_k}
        \|{\bDelta_k}\|_\infty =\|{[\bDelta_k]}_{\mathbb S_k}\|_\infty \leq  r
    \end{equation}

By Lemma \ref{lemma:Sk}, we have, 
\begin{equation*}
    \|\W_k\|_\infty = O_p\left(\sqrt{\frac{p_k^3\log p_k}{n (p^*)^2}}\right)+O\left(\frac{1}{\sqrt{p^*}}\right).
\end{equation*}
Combining Equation \eqref{eq:Delta_k}, we get 
\begin{equation*}
    \| \breve{M}_k-M_k\|_\infty\leq (p^*/p_k)\left\{O_p\left(\sqrt{\frac{p_k^3\log p_k}{n (p^*)^2}}\right)+O\left(\frac{1}{\sqrt{p^*}}\right)\right\}.
\end{equation*}

We next show that the strict dual feasibility $\breve{M}_k=\hat{M}_k$ holds. It sufficient to prove that  ${[\mathbf Z_k]}_{\mathbb S^c}:=\frac{1}{p_k\lambda_k}(-[\hat{\S}_k]_{\mathbb S^c}+[\breve{M}_k^{-1}]_{\mathbb S^c})$ satisfies $\|{[\mathbf Z_k]}_{\mathbb S^c}\|_\infty<1$.

We rewrite $G([\breve{M}_k]_{\mathbb S_k})=0$ as 
\begin{equation*}
    M_k^{-1}\bDelta_k M_k^{-1}+\W_k-R_k(\bDelta_k)+p_k\lambda_k\mathbf Z_k=0.
\end{equation*}

Then we vectorize the above equation, we have,
\begin{equation*}
    \bGamma^*_k\vec(\bDelta_k)+\vec(\W_k)-\vec(R_k(\bDelta_k))+p_k\lambda_k\vec(\mathbf Z_k)=0.
\end{equation*}
where $\breve\bGamma^*_k:=M_k^{-1}\otimes M_k^{-1}$. Since $[\bDelta_k]_{\mathbb S_k^c}=0$, we separate it as two parts,
\begin{equation}\label{eq:G1}
[\bGamma^*_k]_{\mathbb S_k,\mathbb S_k}\vec([\bDelta_k]_{\mathbb S_k})+\vec([\W_k]_{\mathbb S_k})-\vec([R_k(\bDelta_k)]_{\mathbb S_k})+p_k\lambda_k\vec([\mathbf Z_k]_{\mathbb S_k})=0,
\end{equation}
and
\begin{equation}\label{eq:G2}
        [\breve\bGamma^*_k]_{\mathbb S_k^c,\mathbb S_k}\vec([\bDelta_k]_{\mathbb S_k})+\vec([\W_k]_{\mathbb S_k^c})-\vec([R_k(\bDelta_k)]_{\mathbb S_k^c})+p_k\lambda_k\vec([\mathbf Z_k]_{\mathbb S_k^c})=0.
\end{equation}

We rewrite Equation \eqref{eq:G1} as 
\begin{equation*}
    \vec([\bDelta_k]_{\mathbb S_k})=([\breve\bGamma^*_k]_{\mathbb S_k,\mathbb S_k})^{-1} \left\{ -\vec([\W_k]_{\mathbb S_k})+\vec([R_k(\bDelta_k)]_{\mathbb S_k})-p_k\lambda_k\vec([\mathbf Z_k]_{\mathbb S_k})\right\},
\end{equation*}
and then plug it in Equation \eqref{eq:G2},
\begin{equation*}
    \begin{aligned}
        \vec([\mathbf Z_k]_{\mathbb S_k^c})=&-\frac{1}{p_k\lambda_k}[\breve\bGamma^*_k]_{\mathbb S_k^c,\mathbb S_k}\vec([\bDelta_k]_{\mathbb S_k})+\frac{1}{p_k\lambda_k}\vec([\W_k]_{\mathbb S_k^c})-\frac{1}{p_k\lambda_k}\vec([R_k(\bDelta_k)]_{\mathbb S_k^c})\\
        =&\frac{1}{p_k\lambda_k}[\bGamma^*_k]_{\mathbb S_k^c,\mathbb S_k}([\bGamma^*_k]_{\mathbb S_k,\mathbb S_k})^{-1} \left\{ \vec([\W_k]_{\mathbb S_k})-\vec([R_k(\bDelta_k)]_{\mathbb S_k})+p_k\lambda_k\vec([\mathbf Z_k]_{\mathbb S_k})\right\}\\
        &~~~~~~~~+\frac{1}{p_k\lambda_k}\vec([\W_k]_{\mathbb S_k^c})-\frac{1}{p_k\lambda_k}\vec([R_k(\bDelta_k)]_{\mathbb S_k^c}).
    \end{aligned}
\end{equation*}
where the last equation holds by the definition of $M_k$.

Taking $l_\infty$ norm, we have,
\begin{equation*}
    \begin{aligned}
        \|\vec([\mathbf Z_k]_{\mathbb S_k^c})\|_\infty
        \leq&\frac{1}{p_k\lambda_k}\|[\bGamma^*_k]_{\mathbb S_k^c,\mathbb S_k}([\bGamma^*_k]_{\mathbb S_k,\mathbb S_k})^{-1} \|_{L_{\infty}}\left\{ \|\vec([\W_k]_{\mathbb S_k})\|_{\infty}+\|\vec([R_k(\bDelta_k)]_{\mathbb S_k})\|_{\infty}\right\}\\
        &+\|[\bGamma^*_k]_{\mathbb S_k^c,\mathbb S_k}([\bGamma^*_k]_{\mathbb S_k,\mathbb S_k})^{-1} \|_{L_{\infty}}\|\vec([\mathbf Z_k]_{\mathbb S_k})\|_{\infty}\\
        &+\frac{1}{p_k\lambda_k}\|\vec([\W_k]_{\mathbb S_k^c})\|_{\infty}+\frac{1}{p_k\lambda_k}\|\vec([R_k(\bDelta_k)]_{\mathbb S_k^c})\|_{\infty}\\
        \leq &\frac{2-\alpha_k}{p_k\lambda_k}\left\{ \|\W_k\|_\infty+\|R_k(\bDelta_k)\|_\infty \right\}+1-\alpha_k.
    \end{aligned}
\end{equation*}
 where the last inequality holds by Assumption \ref{ass:bounded_complexity} and Equation \eqref{eq:Z_k}.

Since $r\leq (p^*/p_k)^{1/2}/\{3 d_k(C_3)^4(1+8/\alpha_k)\}$, we know that, $\|R_k(\bDelta_k)\|_\infty\leq 3/2 (C_3)^3d_k(p^*/p_k)^{-3/2}r^2\leq p_k\lambda_k\alpha_k/8$, we have,
\begin{equation*}
    \|[\mathbf Z_k]_{\mathbb S_k^c}\|_\infty\leq \frac{2-\alpha_k}{p_k\lambda_k}\frac{2\alpha_k p_k\lambda_k}{8}+1-\alpha_k\leq 1-\alpha_k/2<1,
\end{equation*}
with high probability. Therefore, we prove that $\hat{M}_k=\breve{M}_k$. Thus, we have,
\begin{equation*}
    \| \hat{M}_k-M_k\|_\infty\leq (p^*/p_k)\left\{O_p\left(\sqrt{\frac{p_k^3\log p_k}{n (p^*)^2}}\right)+O\left(\frac{1}{\sqrt{p^*}}\right)\right\}.
\end{equation*}
Furthermore, since the support of $\hat{M}_k$ is the same as that of $M_k$, we have,
\begin{equation*}
    \| \hat{M}_k-M_k\|_2\leq \| \hat{M}_k-M_k\|_{L_1}\leq d_k(p^*/p_k)\left\{O_p\left(\sqrt{\frac{p_k^3\log p_k}{n (p^*)^2}}\right)+O\left(\frac{1}{\sqrt{p^*}}\right)\right\}.
\end{equation*}

Then, for $\hat{\bOmega}_k$, by triangle inequality,
\begin{equation*}
    \begin{aligned}
        \|\hat{\bOmega}_k-\bOmega_k^*\|_\infty=&\left\|\frac{\hat{M}_k}{\|\hat{M}_k\|_F}-\frac{M_k}{\|M_k\|_F}\right\|_\infty\\
        \leq&\left\|\frac{\hat{M}_k}{\|\hat{M}_k\|_F}-\frac{M_k}{\|\hat{M}_k\|_F}\right\|_\infty+\left\|\frac{M_k}{\|\hat{M}_k\|_F}-\frac{M_k}{\|M_k\|_F}\right\|_\infty\\
        \leq &\frac{1}{\|\hat{M}_k\|_F}\left\| \hat{M}_k-M_k\right\|_\infty+\frac{\|{M}_k\|_\infty}{\|{M}_k\|_F\|\hat{M}_k\|_F}\left\| \hat{M}_k-M_k\right\|_F\\
        \leq &\frac{1}{(p^*)^{1/2}}(p^*/p_k)\left\{O_p\left(\sqrt{\frac{p_k^3\log p_k}{n (p^*)^2}}\right)+O\left(\frac{1}{\sqrt{p^*}}\right)\right\}\\
        &~~~~~~~~+\frac{(p^*/p_k)^{1/2}}{(p^*)^{1/2}}\left\{O_p\left(\sqrt{\frac{(p_k+s_k)p_k\log p_k}{n p^*}}\right)+O\left(\sqrt{\frac{p_k+s_k}{p_k^2}}\right)\right\}\\
        =&O_p\left(\sqrt{\frac{p_k\log p_k}{n p^*}}\right)+O\left(\frac{1}{p_k}\right).\\
    \end{aligned}
\end{equation*}
and
\begin{equation*}
    \begin{aligned}
        \|\hat{\bOmega}_k-\bOmega_k^*\|_2=&\left\|\frac{\hat{M}_k}{\|\hat{M}_k\|_F}-\frac{M_k}{\|M_k\|_F}\right\|_2\\
        \leq&\left\|\frac{\hat{M}_k}{\|\hat{M}_k\|_F}-\frac{M_k}{\|\hat{M}_k\|_F}\right\|_2+\left\|\frac{M_k}{\|\hat{M}_k\|_F}-\frac{M_k}{\|M_k\|_F}\right\|_2\\
        \leq &\frac{1}{\|\hat{M}_k\|_F}\left\| \hat{M}_k-M_k\right\|_2+\frac{\|{M}_k\|_2}{\|{M}_k\|_F\|\hat{M}_k\|_F}\left\| \hat{M}_k-M_k\right\|_F\\
        \leq &\frac{1}{(p^*)^{1/2}}d_k(p^*/p_k)\left\{O_p\left(\sqrt{\frac{p_k^3\log p_k}{n (p^*)^2}}\right)+O\left(\frac{1}{\sqrt{p^*}}\right)\right\}\\
        &~~~~~~~~+\frac{(p^*/p_k)^{1/2}}{(p^*)^{1/2}}\left\{O_p\left(\sqrt{\frac{(p_k+s_k)p_k\log p_k}{n p^*}}\right)+O\left(\sqrt{\frac{p_k+s_k}{p_k^2}}\right)\right\}\\
        =&d_k\left\{O_p\left(\sqrt{\frac{p_k\log p_k}{n p^*}}\right)+O\left(\frac{1}{p_k}\right)\right\}.\\
    \end{aligned}
\end{equation*}

\underline{Proof of Equation \eqref{eq:FBr}:}

  For ${[\bDelta_k^\prime]}_{\mathbb S_k}\in B(r)$, since ${[\bDelta_k]}_{\mathbb S_k^c}=0$, we have
    \begin{align*}
        &~~~~F(\Vec({[\bDelta_k^\prime]}_{\mathbb S_k}))\\
        &=-(p^*/p_k)({\bGamma}^*_{k,\mathbb S_k,\mathbb S_k})^{-1}{\Vec}(G([M_k]_{\mathbb S_k}+{[\bDelta_k^\prime]}_{\mathbb S_k}))+{\Vec}({[\bDelta_k^\prime]}_{\mathbb S_k})\\
        &=-(p^*/p_k)({\bGamma}^*_{k,\mathbb S_k,\mathbb S_k})^{-1}{\Vec}(-[(M_k+\bDelta_k^\prime)^{-1}-M_k^{-1}]_{\mathbb S_k}+{[{\bf W}_k]}_{\mathbb S_k}+p_k\lambda_k {[{\bf Z}_k]}_{\mathbb S_k})+{\Vec}([\bDelta_k^\prime]_{\mathbb S_k})\\
        &=-(p^*/p_k)({\bGamma}^*_{k,\mathbb S_k,\mathbb S_k})^{-1}\left\{{\Vec}(-[(M_k+\bDelta_k^\prime)^{-1}-M_k^{-1}]_{\mathbb S_k})+({\bGamma}^*_{k,\mathbb S_k,\mathbb S_k})\Vec([\bDelta_k^\prime]_{\mathbb S_k})\right\}\\
        &~~~~~~~~~~~~~-(p^*/p_k)({\bGamma}^*_{k,\mathbb S_k,\mathbb S_k})^{-1}\Vec({[\bf W_k]}_{\mathbb S_k}+p_k\lambda_k {[{\bf Z}_k]}_{\mathbb S_k})\\
        &=-(p^*/p_k)({\bGamma}^*_{k,\mathbb S_k,\mathbb S_k})^{-1}{\Vec}(\left[-\{(M_k+\bDelta_k^\prime)^{-1}-M_k^{-1}\}+M_k^{-1}\bDelta_k^\prime M_k^{-1}\right]_{\mathbb S_k})\\
        &~~~~~~~~~~~~~-(p^*/p_k)({\bGamma}^*_{k,\mathbb S_k,\mathbb S_k})^{-1}\Vec({[\bf W_k]}_{\mathbb S_k}+p_k\lambda_k {[\bf Z]}_{\mathbb S_k})\\
        &= (p^*/p_k)({\bGamma}^*_{k,\mathbb S_k,\mathbb S_k})^{-1} \Vec([R(\bDelta_k^\prime)]_{\mathbb S_k})-(p^*/p_k)({\bGamma}^*_{k,\mathbb S_k,\mathbb S_k})^{-1}\Vec({[\bf W_k]}_{\mathbb S_k}+\lambda_k {[\bf Z_k]}_{\mathbb S_k})
    \end{align*}
    Hence, since $r\leq \min\{\frac{(p^*/p_k)^{1/2}}{3(C_3)^4 d_k}, \frac{(p^*/p_k)^{1/2}}{3 C_3 d_k}\}$, which can be proved by Assumption \ref{ass:bounded_complexity}
    \begin{align*}
    \|F(\Vec([{\bDelta_k}^\prime]_{\mathbb S_k}))\|_\infty
    &\leq C_3(p^*/p_k)\|R({\bDelta_k}^\prime)\|_\infty+C_3(p^*/p_k)(\|{\bf W}_k\|_\infty +p_k\lambda_k)\\
    &\leq 3/2 (C_3)^4 d_k(p^*/p_k)^{-1/2}\|\bDelta_k^\prime\|_\infty^2+r/2\\
    &\leq 3/2 (C_3)^4 d_k(p^*/p_k)^{-1/2}r^2+r/2\\
    &\leq r.
    \end{align*}

The proof is completed.
\subsection{Proof of Theorem \ref{thm:selection}}

By Theorem \ref{thm:estimation_error_spectral}, we see that 
\begin{equation*}
    \Vert \hat{\bOmega}_k-\bOmega_k\Vert_\infty=O_p\left(\sqrt{\frac{p_k\log p_k}{n p^*}}\right)+O\left(\frac{1}{p_k}\right),
\end{equation*}
that is, with the probability larger than $1-\beta$,
\begin{equation*}
    \Vert \hat{\bOmega}_k-\bOmega_k\Vert_\infty<\tau_k,
\end{equation*}

For $(i,j)\notin \mathcal{S}(\bOmega_k^*)$, we have $\vert[\breve{\bOmega}_k]_{i,j}\vert=\vert[\hat{\bOmega}_k]_{i,j}-[{\bOmega}_k^*]_{i,j}\vert\leq \tau_k$, which implies $[\breve{\bOmega}_k]_{i,j}=0$. For $(i,j)\in \mathcal{S}(\bOmega_k^*)$, since $\theta_{\min}>2\max_{1\leq k\leq  K}\tau_k$, we have, $\vert[\hat{\bOmega}_k]_{i,j}\vert\geq \vert [{\bOmega}_k]_{i,j}^*\vert-\vert [\hat{\bOmega}_k]_{i,j}-[{\bOmega}_k]_{i,j}^*\vert>\theta_{\min}-\tau_k>\tau_k$, which implies $[\breve{\bOmega}_k]_{i,j}=[\hat{\bOmega}_k]_{i,j}$. When $[{\bOmega}_k]_{i,j}^*>0$, $[\breve{\bOmega}_k]_{i,j}=[\hat{\bOmega}_k]_{i,j}\geq  [{\bOmega}_k]_{i,j}^*-\vert [\hat{\bOmega}_k]_{i,j}-[{\bOmega}_k]_{i,j}^*\vert>\tau_k$, which implies $\text{sign}([{\bOmega}_k]_{i,j}^*)=\text{sign}([{\breve\bOmega}_k]_{i,j})$. Similarly, when $[{\bOmega}_k]_{i,j}^*<0$, $[\breve{\bOmega}_k]_{i,j}=[\hat{\bOmega}_k]_{i,j}\leq  [{\bOmega}_k]_{i,j}^*+\vert [\hat{\bOmega}_k]_{i,j}-[{\bOmega}_k]_{i,j}^*\vert<-\tau_k$, which implies $\text{sign}([{\bOmega}_k]_{i,j}^*)=\text{sign}([{\breve\bOmega}_k]_{i,j})$. The proof is completed.

\subsection{Preliminary lemmas}

\begin{lemma}{(General moments of spherically symmetric distribution)}\label{lemma:moments}Suppose ${\boldsymbol u}=(u_1,u_2,\cdots,u_p)^\top\in\mathbb R^{p}$ is uniformly distributed on $\mathbb{S}^{p-1}$, then for any integers $m_1, \ldots, m_p$, with $m = \sum_{i=1}^p m_i$, the mixed moments of $\boldsymbol u$ can be expressed as:
\[
\E \left( \prod_{i=1}^p u_i^{m_i} \right) =
\begin{cases}
\frac{1}{(p/2)^{[l]}} \frac{(2l_i)!}{\prod_{i=1}^n 4^{l_i}(l_i)!}, & \text{if } m_i = 2l_i \text{ are even}, \, i = 1, \ldots, p, \, m = 2l; \\
0, & \text{if at least one of the } m_i \text{ is odd.}
\end{cases}
\]
where as above $x^{[l]} = x(x+1)\cdots(x+l-1)$.

Specifically,
\[\E(u_i u_j u_k u_l)=\frac{\delta_{i j}\delta_{kl}+\delta_{i k}\delta_{j l}+\delta_{i l}\delta_{j k}}{p(p+2)}\]
where $\delta_{ij}=\mathbb I(i=j)$.
\end{lemma}
\begin{proof}
  See the Theorem 3.3 in Section 3.1.2. of \cite{fang-2018-symmetric}  
\end{proof}

Before the Lemma 2, we give the definition of Lipschitz function and norm.
\begin{definition}[Lipschitz functions]
\label{def:lipschitz}
Let $(X, d_X)$ and $(Y, d_Y)$ be metric spaces.  
A function $f \colon X \to Y$ is called \emph{Lipschitz} if there exists $L \in \mathbb{R}$ such that  
\[
d_Y(f(u), f(v)) \leq L \cdot d_X(u, v) \quad \text{for every } u, v \in X.
\]
The infimum of all $L$ in this definition is called the \emph{Lipschitz norm} of $f$ and is denoted $\|f\|_{\text{Lip}}$.
\end{definition}

The property for spherical distributions are shown as follows.
\begin{lemma}
    \label{lemma:lipschitz-concentration}
Consider a random vector \( \boldsymbol{X} \sim \sqrt{p}\text{Unif}(\mathbb{S}^{p-1})\), i.e.\ \( \boldsymbol{X} \) is uniformly distributed on the Euclidean sphere of radius \(\sqrt{p}\). Consider a Lipschitz function\footnote{Here we mean a function satisfying the Lipschitz condition with respect to the Euclidean distance on the sphere.} \( f : \sqrt{p}\mathbb{R}^{p} \to \mathbb{R} \). Then
\[
\|f(\boldsymbol{X}) - \E \{f(\boldsymbol{X})\}\|_{\psi_2} \leq C_{Lip}\|f\|_{\text{Lip}},
\]
for some constant $C_{Lip}$.
\end{lemma}
\begin{proof}
    See the Theorem 5.1.4 in Section 5.1.2 of \cite{vershynin2018high}.
\end{proof}

\begin{lemma}[Bernstein's inequality]\label{lemma:bernstein}
Let \( X_1, \ldots, X_N \) be independent, mean zero, sub-exponential random variables. Then, for every \( t \geq 0 \), we have
\[
\mathbb{P} \left( \left| \sum_{i=1}^N X_i \right| \geq t \right) 
\leq 2 \exp \left\{ -c \min \left( \frac{t^2}{\sum_{i=1}^N \| X_i \|_{\psi_1}^2}, \frac{t}{\max_i \| X_i \|_{\psi_1}} \right) \right\},
\]
where \( c > 0 \) is an absolute constant.
\end{lemma}
\begin{proof}
    See the Theorem 2.8.1 in Section 2.8 of \cite{vershynin2018high}.
\end{proof}

\begin{lemma}
    Suppose ${\boldsymbol u}=(u_1,\cdots, u_p)^\top\in \mathbb R^{p^*}$ is uniformly distributed on $\mathbb{S}^{p^*-1}$, then under Assumption \ref{ass:eigenvalues}-\ref{ass:lambda_inverse}, we have 
    \[
    \left|\E\left(\frac{u_iu_j}{{\boldsymbol u}^\top {\bSigma^*}{\boldsymbol u}}\right)-\frac{1}{\tr({\bSigma^*})}\delta_{ij}\right|=O\left\{(p^*)^{-\frac{3}{2}}\right\}
    \]
\end{lemma}
\begin{proof}
    By Lemma \ref{lemma:moments} and Assumption \ref{ass:eigenvalues}, we have,
    \begin{equation*}
        \begin{aligned}
           &\left|\E\left(\frac{u_i u_j}{{\boldsymbol u}^\top {\bSigma^*}{\boldsymbol u}}\right)-\frac{1}{\tr({\bSigma^*})}\delta_{ij}\right|\\
           \leq &\E\left|\frac{u_i u_j}{{\boldsymbol u}^\top {\bSigma^*}{\boldsymbol u}}-\frac{u_i u_j}{\E({\boldsymbol u}^\top {\bSigma^*}{\boldsymbol u})}\right|
           = \E\left|\frac{u_i u_j\left\{{\boldsymbol u}^\top {\bSigma^*}{\boldsymbol u}-\E({\boldsymbol u}^\top {\bSigma^*}{\boldsymbol u})\right\}}{{\boldsymbol u}^\top {\bSigma^*}{\boldsymbol u}\E({\boldsymbol u}^\top {\bSigma^*}{\boldsymbol u})}\right|\\
           \leq &\frac{1}{C_1^K}\sqrt{\E\left(\frac{u_i u_j}{{\boldsymbol u}^\top {\bSigma^*}{\boldsymbol u}}\right)^2\var({\boldsymbol u}^\top {\bSigma^*}{\boldsymbol u})}\\
    \end{aligned}
    \end{equation*}
By Cauchy equality, we have $({\boldsymbol u}^\top \bSigma^*{\boldsymbol u} )({\boldsymbol u}^\top(\bSigma^*)^{-1}{\boldsymbol u} )\geq 1$. Then by Assumption \ref{ass:lambda_inverse},
\begin{equation*}
    \begin{aligned}
        \E\left(\frac{u_iu_j}{{\boldsymbol u}^\top \bSigma^*{\boldsymbol u}}\right)^2
        \leq &\E\left(u_iu_j{\boldsymbol u}^\top{(\bSigma^*)}^{-1}{\boldsymbol u}\right)^2=\E\left(\sum_{k_1,l_1}\sum_{k_2,l_2}[(\bSigma^*)^{-1}]_{k_1,l_1}[(\bSigma^*)^{-1}]_{k_2,l_2}u_i^2u_j^2u_{k_1}u_{k_2}u_{l_1}u_{l_2}\right)\\
        = & \frac{105(p^*)^2}{\{\tr(\bSigma)\}^2}\cdot\frac{\sum_{k_1,l_1=1}^{p^*}\sum_{k_2,l_2=1}^{p^*}v_{k_1l_1}v_{k_2l_2}}{p^*(p^*+2)(p^*+4)(p^*+6)}\\
        \leq &\frac{105}{C_1^{2K}}\cdot\frac{{(p^*)}^2}{p^*(p^*+2)(p^*+4)(p^*+6)}\\
        =& O\left\{(p^*)^{-2}\right\}.
    \end{aligned}
\end{equation*}

    By Lemma \ref{lemma:moments} and Assumption \ref{ass:eigenvalues}, we have
    \begin{equation}\label{eq:varuu}
            \begin{aligned}
        \var(\boldsymbol u^\top \bSigma^*\boldsymbol u)
        =&\E({\boldsymbol u}^\top {\bSigma^*}{\boldsymbol u}{\boldsymbol u}^\top {\bSigma^*}{\boldsymbol u})-\{\E({\boldsymbol u}^\top {\bSigma^*}{\boldsymbol u})\}^2\\
         =&\frac{\tr\{\tr(\bSigma^*)\bSigma^*+2(\bSigma^*)^2\}}{p^*(p^*+2)}-\frac{\{\tr(\bSigma^*)\}^2}{(p^*)^2}\\
         =&\frac{2\tr((\bSigma^*)^2)}{p^*(p^*+2)}+2\frac{\{\tr(\bSigma^*)\}^2}{(p^*)^2(p^*+2)}\\
         \leq& \frac{2\sum_{i=1}^{p^*}\lambda_i^2(\bSigma^*)}{p^*(p^*+2)}+\frac{2}{C_1^{2K}(p^*+2)}\\
        \leq& \frac{2}{C_1^{2K}(p^*+2)}+\frac{2}{C_1^{2K}(p^*+2)}\\
        =&O\left\{(p^*)^{-1}\right\}
    \end{aligned}
\end{equation}

    A hidden calculation involved here is
    \begin{equation*}
            \begin{aligned}
          \E[{\boldsymbol u}^\top \bSigma^*{\boldsymbol u} {\boldsymbol u}{\boldsymbol u}^\top]_{ij}
          &=\sum_{k,l=1}^{p^*}[\bSigma^*]_{k,l}(\delta_{ij}\delta_{kl}+\delta_{ik}\delta_{jl}+\delta_{il}\delta_{jk})\\
          &=\{\tr(\bSigma^*)+[\bSigma^*]_{i,i}+[\bSigma^*]_{j,j})\delta_{ij}+([\bSigma^*]_{i,j}+[\bSigma^*]_{j,i}\}(1-\delta_{ij}).
    \end{aligned}
\end{equation*}

    Combine them together, the proof is finished.
\end{proof}

 For each $k=1,\ldots,K$, we define $\mathbb B(\bOmega_k^*)$ as the set containing $\bOmega_k^*$ and its neighborhood for some sufficiently large constant radius $\alpha>0$, i.e. 
\begin{equation*}
    \begin{aligned}
        \mathbb B(\bOmega_k^*):=\left\{ \bOmega\in \mathbb R^{p_k\times p_k}:\bOmega=\bOmega^\top ;\bOmega\succ 0;\Vert \bOmega-\bOmega_k^*\Vert_F\leq \alpha\right\}.
    \end{aligned}
\end{equation*}
\begin{lemma}\label{lemma:Sk}
Suppose the tensor data $\T_1,\ldots,\T_n\in \mathbb R^{p_1\times\cdots\times p_K}$ are sample from a tensor elliptical distribution $\mathcal{E}_{\boldsymbol p}(\bmu,\bSigma,\psi)$ and Assumption \ref{ass:eigenvalues}-\ref{ass:lambda_inverse} hold. For any fixed $\bOmega_j\in \mathbb B(\bOmega_k^*)$ with $\|\bOmega_j\|_2\lesssim p_j^{-1/2},j\neq k$, we have 
\begin{equation}\label{lemmaeq:Elambda_rate}
    \left\Vert \E(\S_k)-\frac{p_k\prod_{j\neq k}\tr(\bLambda_j^* \bOmega_j)}{p^*}\bLambda_k^* \right\Vert_\infty \leq
    O\left(\frac{1}{\sqrt{p^*}}\right),
\end{equation}
where $\S_k=(S_{k,jl})_{j,l\in\{1,\ldots,p_k\}}=p_k n^{-1}\sum_{i=1}^n [U(\T_i-\bmu)]_{(k)}(\bOmega_K\otimes\cdots\otimes\bOmega_{k+1}\otimes\bOmega_{k-1}\otimes\cdots\otimes\bOmega_1)[U(\T_i-\bmu)]_{(k)}^\top$. Moreover, we have
\begin{equation}\label{lemmaeq:Sk_rate}
    \left\Vert \S_k -\frac{p_k\prod_{j\neq k}\tr(\bLambda_j^* \bOmega_j)}{p^*}\bLambda_k^*\right\Vert_\infty=O_p\left(\sqrt{\frac{p_k^3\log p_k}{n (p^*)^2}}\right)+O\left(\frac{1}{\sqrt{p^*}}\right).
\end{equation}
\end{lemma}

\begin{proof}

Note that, by the property of tensor elliptical distribution, $[\T_i-\bmu]_{(k)}\in \mathbb R^{p_k\times(p^*/p_k)}\sim TE(\boldsymbol 0,\bSigma_1^*,\bSigma_K^*\otimes\cdots\otimes\bSigma_{k+1}^*\otimes\bSigma_{k-1}^*\otimes\cdots\otimes\bSigma_1^*)$, that is,
\begin{equation*}
    [\T_i-\bmu]_{(k)}=v_i(\bSigma_k^*)^{1/2}\mathbf u_i (\bSigma_K^*\otimes\cdots\otimes\bSigma_{k+1}^*\otimes\bSigma_{k-1}^*\otimes\cdots\otimes\bSigma_1^*)^{1/2},
\end{equation*}
where $\u_i:=\vec(\mathbf u_i)$ is a random vector distributed uniformly on the unit sphere surface in $\mathbb R^{p^*}$.

By definition and Lemma \ref{lemma:moments},
    \begin{equation*}
        \begin{aligned}
            &p_k^{-1}\left|\E([\S_{k}]_{j,l})-\E\left\{ \frac{v_i^{-2}\e_j^\top[\T_i-\bmu]_{(k)}(\bOmega_K\otimes\cdots\otimes\bOmega_{k+1}\otimes\bOmega_{k-1}\otimes\cdots\otimes\bOmega_1)[\T_i-\bmu]_{(k)}^\top\e_l  }{\E \left(v_i^{-2}\Vert \T_i-\bmu\Vert^{2}\right)}\right\}\right|\\
            =&\left|\E\left\{\e_j^\top[U(\T_i-\bmu)]_{(k)}\bOmega_{K-k}[U(\T_i-\bmu)]_{(k)}^\top\e_l  \right\}-\E \left\{\frac{v_i^{-2}\e_j^\top[\T_i-\bmu]_{(k)}\bOmega_{K-k}[\T_i-\bmu]_{(k)}^\top\e_l  }{\E \left(v_i^{-2}\Vert \T_i-\bmu\Vert^{2}\right)}\right\}\right|\\
            =&\left|\E\left\{\frac{\left(v_i^{-2}\Vert \T_i-\bmu\Vert^{2}-\E v_i^{-2}\Vert \T_i-\bmu\Vert^{2}\right)v_i^{-2}\e_j^\top[\T_i-\bmu]_{(k)}\bOmega_{K-k}[\T_i-\bmu]_{(k)}^\top\e_l }{v_i^{-2}\Vert \T_i-\bmu\Vert^{2}\E \left(v_i^{-2}\Vert \T_i-\bmu\Vert^{2}\right)} \right\}\right|\\
            \leq &\frac{p^*}{\tr(\bSigma^*)}\sqrt{\var(\u_i^\top \bSigma^* \u_i)\E\left\{\frac{\e_j^\top(\bSigma_k^*)^{1/2}  \mathbf u_i \left(\bSigma^{*1/2}\bOmega\bSigma^{*1/2}\right)_{K-k}\mathbf u_i^\top (\bSigma_k^{*})^{1/2}\e_l}{\u_i^\top \bSigma^*\u_i}\right\}^2}
        \end{aligned}
    \end{equation*}
    where $\bOmega_{K-k}=\bOmega_K\otimes\cdots\otimes\bOmega_{k+1}\otimes\bOmega_{k-1}\otimes\cdots\otimes\bOmega_1$ and $\left(\bSigma^{*1/2}\bOmega\bSigma^{*1/2}\right)_{K-k}=(\bSigma_K^*)^{1/2}\bOmega_K(\bSigma_K^*)^{1/2}\otimes\cdots\otimes(\bSigma_{k+1}^*)^{1/2}\bOmega_{k+1}(\bSigma_{k+1}^*)^{1/2}\otimes(\bSigma_{k-1}^*)^{1/2}\bOmega_{k-1}(\bSigma_{k-1}^*)^{1/2}\otimes\cdots\otimes(\bSigma_1^*)^{1/2}\bOmega_1(\bSigma_1^*)^{1/2}$.

    Similar with Equation \eqref{eq:varuu}, $\var(\u_i^\top \bSigma^* \u_i)=O\left\{(p^{*})^{-1}\right\}$. By Lemma \ref{lemma:moments} and Assumption \ref{ass:eigenvalues}--\ref{ass:lambda_inverse}, 
\begin{small}
    \begin{equation*}
        \begin{aligned}
            &\E\left\{\frac{\e_j^\top(\bSigma_k^*)^{1/2}  \mathbf u_i \left(\bSigma^{*1/2}\bOmega\bSigma^{*1/2}\right)_{K-k}\mathbf u_i^\top (\bSigma_k^{*})^{1/2}\e_l}{\u_i^\top \bSigma^*\u_i}\right\}^2\\
            \leq &\E\left\{\e_j^\top(\bSigma_k^*)^{1/2}  \mathbf u_i \left(\bSigma^{*1/2}\bOmega\bSigma^{*1/2}\right)_{K-k}\mathbf u_i^\top (\bSigma_k^{*})^{1/2}\e_l \u_i^\top (\bSigma^*)^{-1}\u_i\right\}^2\\
            =&\E \left\{ \sum_{s_1,t_1=1}^{p_k}[(\bSigma_k^*)^{1/2}]_{j,s_1}[(\bSigma_k^*)^{1/2}]_{t_1,l}\sum_{s_2,t_2=1}^{p^*/p_k}[\mathbf u_i]_{s_1,s_2} [\left(\bSigma^{*1/2}\bOmega\bSigma^{*1/2}\right)_{K-k}]_{s_2,t_2}[\mathbf u_i]_{t_1,t_2}\sum_{s_3,t_3=1}^{p^*}[\u_{i}]_{s_3} [(\bSigma^*)^{-1}]_{s_3,t_3}[\u_{i}]_{t_3}\right\}^2\\
            \leq &O\{(p^*)^{-4}\}\left\{ \sum_{s_1,t_1=1}^{p_k}[(\bSigma_k^*)^{1/2}]_{j,s_1}[(\bSigma_k^*)^{1/2}]_{t_1,l}\sum_{s_2,t_2=1}^{p^*/p_k}[\left(\bSigma^{*1/2}\bOmega\bSigma^{*1/2}\right)_{K-k}]_{s_2,t_2}\sum_{s_3,t_3=1}^{p^*}[(\bSigma^*)^{-1}]_{s_3,t_3}\right\}^2\\
            \leq&O\{(p^*)^{-4}\}\Vert\bSigma^*\Vert_{L_1}^2\left\{\sum_{s_2,t_2=1}^{p^*/p_k}[\left(\bSigma^{*1/2}\bOmega\bSigma^{*1/2}\right)_{K-k}]_{s_2,t_2}\sum_{s_3,t_3=1}^{p^*}[(\bSigma^*)^{-1}]_{s_3,t_3}\right\}^2\\
            \leq &O((p^*)^{-4})\Vert\bSigma^*\Vert_{L_1}^2\left(\frac{p^*}{p_k}\right)^2 \left\Vert\left(\bSigma^{*1/2}\bOmega\bSigma^{*1/2}\right)_{K-k}\right\Vert_F^2 (p^*)^2\Vert(\bSigma^*)^{-1}\Vert_{L_1}^2\\
            =&O(p_k^{-2})\prod_{l\neq k}^K\tr(\bOmega_l\bSigma_l^*)^2=O(p_k^{-2}).
        \end{aligned}
    \end{equation*}
\end{small}
    Thus, we have,
    \begin{equation}\label{eq:Elambda_rate}
        p_k^{-1}\left|\E([\S_{k}]_{j,l})-\E\left\{ \frac{v_i^{-2}\e_j^\top[\T_i-\bmu]_{(k)}\bOmega_{K-k}[\T_i-\bmu]_{(k)}^\top\e_l  }{\E \left(v_i^{-2}\Vert \T_i-\bmu\Vert^{2}\right)}\right\}\right|\leq O\{p_k^{-3/2}{(\prod_{l\neq k}^K p_l)^{-1/2}}\}\\
    \end{equation}
    We then calculate the term $\E \left\{\frac{v_i^{-2}\e_j^\top[\T_i-\bmu]_{(k)}\bOmega_{K-k}[\T_i-\bmu]_{(k)}^\top\e_l  }{\E \left(v_i^{-2}\Vert \T_i-\bmu\Vert^{2}\right)}\right\}$.
    
For $\E \left(v_i^{-2}\Vert \T_i-\bmu\Vert^{2}\right)$,
\begin{equation*}
    \begin{aligned}
        \E \left(v_i^{-2}\Vert \T_i-\bmu\Vert^{2}\right)=\E(\u_i^\top\bSigma^*\u_i)=\frac{\tr(\bSigma^*)}{p^{*}}.
    \end{aligned}
\end{equation*}
For $\E \left\{v_i^{-2}\e_j^\top[\T_i-\bmu]_{(k)}(\bOmega_K\otimes\cdots\otimes\bOmega_{k+1}\otimes\bOmega_{k-1}\otimes\cdots\otimes\bOmega_1)[\T_i-\bmu]_{(k)}^\top\e_l  \right\}$,
\begin{equation*}
    \begin{aligned}
        &\E \left\{v_i^{-2}\e_j^\top[\T_i-\bmu]_{(k)}(\bOmega_K\otimes\cdots\otimes\bOmega_{k+1}\otimes\bOmega_{k-1}\otimes\cdots\otimes\bOmega_1)[\T_i-\bmu]_{(k)}^\top\e_l  \right\}\\
        =&\E\left\{ \e_j^\top (\bSigma_k^*)^{1/2}\mathbf u_i \left(\bSigma^{*1/2}\bOmega\bSigma^{*1/2}\right)_{K-k} \mathbf u_i^\top (\bSigma_k^{*})^{1/2}\e_l\right\}\\
        =&\frac{\tr\left\{\left(\bSigma^{*1/2}\bOmega\bSigma^{*1/2}\right)_{K-k}\right\}}{p^*}(\bSigma_k^*)_{jl}\\
        =&\frac{\prod_{s\neq k}\tr(\bSigma_s^*\bOmega_s)}{p^*}(\bSigma_k^*)_{jl}.
    \end{aligned}
\end{equation*}
Combing Equation \eqref{eq:Elambda_rate}, the proof for Equation \eqref{lemmaeq:Elambda_rate} is completed.

For Equation \eqref{lemmaeq:Sk_rate}, we define
\begin{equation*}
    f_{j\ell}(\u_i):=\frac{\e_j^\top(\bSigma_k^*)^{1/2}  \mathbf u_i \left(\bSigma^{*1/2}\bOmega\bSigma^{*1/2}\right)_{K-k}\mathbf u_i^\top (\bSigma_k^{*})^{1/2}\e_{\ell}}{\u_i^\top \bSigma^*\u_i},
\end{equation*}
and its  {derivative} function is 
\begin{equation*}
\begin{aligned}
        \nabla f_{j\ell}(\u_i)=&\frac{2\u_i^\top \bSigma^*\u_i\cdot \left(\bSigma^{*1/2}\bOmega\bSigma^{*1/2}\right)_{K-k}\otimes (\{(\bSigma_k^{*})^{1/2}_{j\cdot}\}^\top(\bSigma_k^{*})^{1/2}_{\ell\cdot})}{(\u_i^\top \bSigma^*\u_i)^2}\\
        &-\frac{2\e_j^\top(\bSigma_k^*)^{1/2}  \mathbf u_i \left(\bSigma^{*1/2}\bOmega\bSigma^{*1/2}\right)_{K-k}\mathbf u_i^\top (\bSigma_k^{*})^{1/2}\e_{\ell}\bSigma^*\u_i}{(\u_i^\top \bSigma^*\u_i)^2},
\end{aligned}
\end{equation*}
with 
\begin{equation*}
    \begin{aligned}
        \left\|\nabla f_{j\ell}(\u_i)\right\|_2\leq&\frac{2 \left\|\left(\bSigma^{*1/2}\bOmega\bSigma^{*1/2}\right)_{K-k}\right\|_2\left\vert[\bSigma_k^*]_{j,j}\right\vert\left\vert[\bSigma_k^*]_{\ell,\ell}\right\vert}{\u_i^\top \bSigma^*\u_i}\\
        &+\frac{2\left\|\left(\bSigma^{*1/2}\bOmega\bSigma^{*1/2}\right)_{K-k}\right\|_2\left\vert[\bSigma_k^*]_{j,j}\right\vert\left\vert[\bSigma_k^*]_{\ell,\ell}\right\vert\left\|\bSigma^*\u_i\right\|_2}{(\u_i^\top \bSigma^*\u_i)^2}\\
        &\leq 2\left[\{\lambda_{\min}(\bSigma^*)\}^{-1}+\{\lambda_{\min}(\bSigma^*)\}^{-2}\lambda_{\max}(\bSigma^*)\right]\prod_{s=1,s\neq k}^K\left(\|\bSigma_s\|_2\|\bOmega_s\|_2\|\bSigma^*_k\|_\infty^2\right)\\
        &\lesssim \prod_{s=1,s\neq k}^K\|\bOmega_s\|_2, 
    \end{aligned}
\end{equation*}
since $\lambda_{\min}(\bSigma^*)\leq \u_i^\top \bSigma^*\u_i\leq \lambda_{\max}(\bSigma^*)$, $\left\|\bSigma^*\u_i\right\|_2\leq \lambda_{\max}(\bSigma^*)$ 
and $$\left\vert\e_j^\top(\bSigma_k^*)^{1/2}  \mathbf u_i \left(\bSigma^{*1/2}\bOmega\bSigma^{*1/2}\right)_{K-k}\mathbf u_i^\top (\bSigma_k^{*})^{1/2}\e_{\ell}\right\vert\leq \left\|\left(\bSigma^{*1/2}\bOmega\bSigma^{*1/2}\right)_{K-k}\right\|_2\left\vert[\bSigma_k^*]_{j,j}\right\vert\left\vert[\bSigma_k^*]_{\ell,\ell}\right\vert.$$ 

By Lemma \ref{lemma:lipschitz-concentration}, we have
\begin{equation*}
    \begin{aligned}
       &\left\|  \e_j^\top[U(\T_i-\bmu)]_{(k)}(\bOmega_{K-k})[U(\T_i-\bmu)]_{(k)}^\top\e_l-p_k^{-1}\E ([\S_{k}]_{j\ell})\right\|_{\psi_2}\\
      \lesssim &(p^*)^{-1/2}\|f\|_{Lip}\leq (p^*)^{-1/2}\sup_{\u_i}\left\|\nabla f_{j\ell}(\u_i)\right\|_2\lesssim (p^*)^{-1/2}\prod_{s=1,s\neq k}^K\|\bOmega_s\|_2.
    \end{aligned}
\end{equation*}
By Cauchy inequality,
\begin{equation*}
    \begin{aligned}
        &\left\|  p_k\e_j^\top[U(\T_i-\bmu)]_{(k)}(\bOmega_{K-k})[U(\T_i-\bmu)]_{(k)}^\top\e_l-\E ([\S_{k}]_{j,\ell})\right\|_{\psi_1}\\
        \leq &2 p_k\left\|  \e_j^\top[U(\T_i-\bmu)]_{(k)}(\bOmega_{K-k})[U(\T_i-\bmu)]_{(k)}^\top\e_l-p_k^{-1}\E ([\S_{k}]_{j,\ell})\right\|_{\psi_2}\\
        \lesssim &p_k(p^*)^{-1/2}\prod_{s=1,s\neq k}^K\|\bOmega_s\|_2.
    \end{aligned}
\end{equation*}

By Lemma \ref{lemma:bernstein} (Berstein’s inequality), as $t\lesssim n p_k(p^*)^{-1/2}\prod_{s=1,s\neq k}^K\|\bOmega_s\|_2 $, we have
\begin{equation*}
    \begin{aligned}
        &\pr\left(\left|p_k n^{-1}\sum_{i=1}^n \e_j^\top[U(\T_i-\bmu)]_{(k)}(\bOmega_{K-k})[U(\T_i-\bmu)]_{(k)}^\top\e_\ell-\E([S_{k}]_{j,\ell})\right|\geq t \right)\\
        \leq & 2\exp\left\{ -c\frac{n t^2}{\left\{p_k(p^*)^{-1/2}\prod_{s=1,s\neq k}^K\|\bOmega_s\|_2\right\}^2}\right\},
    \end{aligned}
\end{equation*}
for some constant $c>0$.
Then we have, as $t\lesssim n p_k(p^*)^{-1/2}\prod_{s=1,s\neq k}^K\|\bOmega_s\|_2\lesssim n p_k^{3/2}(p^*)^{-1}$, 
\begin{equation*}
    \begin{aligned}
        \pr(\Vert \S_k-\E(\S_k)\Vert_\infty\geq t)\leq 2p_k^2 \exp\left\{ -c\frac{n t^2}{\left\{p_k(p^*)^{-1/2}\prod_{s=1,s\neq k}^K\|\bOmega_s\|_2\right\}^2}\right\},
    \end{aligned}
\end{equation*}

Therefore, we have,
\begin{equation*}
    \Vert \S_k-\E(\S_k)\Vert_\infty =O_p\left(\sqrt{\frac{p_k^3\log p_k}{n (p^*)^2}}\right).
\end{equation*}

Combing the Equation \eqref{lemmaeq:Elambda_rate}, the proof is completed.
\end{proof}

Recall that
\begin{equation*}
    M_k:=\argmin_{\bOmega_k} q(\bOmega_1,\ldots,\bOmega_K)=\frac{p^*}{p_k\prod_{j\neq k}\tr(\bLambda_j^*\bOmega_j)}(\bLambda_k^*)^{-1},
\end{equation*}
and
\begin{equation*}
    \breve{M}_k=\argmin_{\bOmega_k\succ 0,(\bOmega_k)_{\mathbb S^c}=0} L(\bOmega_k),   
\end{equation*}
 $\bDelta_k=\breve{M}_k-M_k$, $\W_k= \S_k-\frac{p_k\prod_{j\neq k}\tr(\bLambda_j^*\bOmega_j)}{p^*}\bLambda_k^*$ and $R_k(\bDelta_k)=(M_k+\bDelta_k)^{-1}-M_k^{-1}+M_k^{-1}\bDelta_k M_k^{-1}$.

\begin{lemma}\label{lemma:lu2025_l11}
$R_k(\bDelta_k)$ can be represented as $M_k^{-1}\bDelta_k M_k^{-1}\bDelta_k\J M_k^{-1}$, where $\J=\sum_{k=0}^\infty (-1)^k (M_k^{-1}\bDelta_k)^k$. Furthermore, under as $C_3 d_k(p^*/p_k)^{-1/2}\|\bDelta_k\|_\infty<1/3$, the element-wise maximum norm of $R_k(\bDelta_k)$ is bounded as
\begin{equation*}
    \|R_k(\bDelta_k)\|_\infty\leq 3/2 (C_3)^3d_k(p^*/p_k)^{-3/2}\|\bDelta_k\|_\infty^2.
\end{equation*}
\end{lemma}
\begin{proof}
    See the Lemma 11 in \cite{lu2025robust}.
\end{proof}

\bibliographystyle{chicago}
\bibliography{ref}
\end{document}